\documentclass[journal]{IEEEtran}
\IEEEoverridecommandlockouts
\usepackage{cite}
\usepackage[hidelinks]{hyperref}
\usepackage{amsmath,amssymb,amsfonts}
\usepackage{algorithm}
\usepackage{algorithmic} 

\usepackage{graphicx}
\usepackage{textcomp}
\usepackage{xcolor}
\usepackage{CJK}
\usepackage{pifont}
\usepackage{amsthm}
\newtheorem{asp}{\indent\bf Assumption}
\newtheorem{thm}{\indent\bf Theorem}

\usepackage{subfig}
\usepackage{float}  %
\usepackage{makecell}
\usepackage{booktabs}

\usepackage{amssymb}
\usepackage{psfrag}
\usepackage{graphicx}
\usepackage{epstopdf}
\usepackage{multirow}
\usepackage{stfloats}
\usepackage{color}
\usepackage[normalem]{ulem}
\usepackage{arydshln}
\usepackage{enumerate}
\usepackage{bbm}
\usepackage{verbatim}
\usepackage{threeparttable}

\begin{document}
\title{Communication-Pipelined Split Federated Learning for Foundation Model Fine-Tuning in UAV Networks
}

\author{
Zizhen Zhou, Ying-Chang Liang, Yanyu Cheng, and Wei Yang Bryan Lim\\

\thanks{

Z. Zhou is with the National Key Laboratory of Wireless Communications, University of Electronic Science and Technology of China (UESTC), Chengdu 611731, China (e-mail: zhouzizhen@std.uestc.edu.cn).

Y.-C. Liang is with the Center for Intelligent Networking and Communications (CINC), University of Electronic Science and Technology of China (UESTC), Chengdu 611731, China (e-mail: liangyc@ieee.org).

Y. Cheng is with the School of Cyberspace, Hangzhou Dianzi University, Hangzhou 310018, China (email: yycheng@hdu.edu.cn).

W. Y. B. Lim is with the College of Computing and Data Science, Nanyang Technological University, Singapore 639798 (e-mail: bryan.limwy@ntu.edu.sg).
}
}

\maketitle

\begin{abstract}
Deploying foundation models (FMs) on uncrewed aerial vehicles (UAVs) promises broad ``low-altitude economy'' applications. 
Split federated learning (SFL)-based fine-tuning leverages distributed data while keeping raw data local and reduces client-side burden by partitioning the model between client and server.  
However, the per-round training latency is dominated by stragglers.
Training paradigms featuring parallel gradient transmission (GT) allocate dedicated portions of downlink communication resources to each client.  
They may leave resources idle and suffer from prolonged GT latency, especially in UAV networks, where the communication latency typically far exceeds the computation latency.
To address this, we propose a sequential GT paradigm, where the server dedicates all downlink resources for the current GT.
We further propose communication-pipelined SFL (CPSFL), characterized by downlink GT priority scheduling and intra-round asynchronous training.
We investigate CPSFL-based LoRA fine-tuning of FMs in UAV networks and formulate an optimization problem to minimize a weighted sum of per-round training latency and worst-case client energy consumption by optimizing the split point selection (SPS) and the computing and communication resource allocation (CCRA) (the uplink bandwidth allocation and the server computing frequency allocation).
To solve this problem, we develop an attention-based deep reinforcement learning (DRL) framework, where the base station agent decides the split point and the CCRA in each round by leveraging previous round information, including UAV trajectories.
Simulation results show that the proposed DRL-based CPSFL scheme outperforms the parallel GT benchmarks, the ablation variants, the fixed CCRA scheme, while approaching the best fixed-SPS scheme.

\end{abstract}

\section{Introduction}
\label{sec_introduction}
Uncrewed aerial vehicles (UAVs) are widely applied in low-altitude economic activities such as natural resource management, power facility inspection, and public security patrol, owing to their flexible deployment and controllable mobility \cite{kheddar2025recent}. 
During missions, UAVs collect large volumes of sensory data (e.g., images and videos), which can be leveraged to adapt foundation models (FMs) to diverse downstream tasks \cite{kheddar2025recent}.  
Fine-tuning FMs, rather than training from scratch, substantially reduces the required data and computation \cite{qu2025mobile}.

Centralized fine-tuning in a data center after the UAVs land is a straightforward approach, but it suffers from some limitations: 
(i) bandwidth-constrained backhaul links from the landing site to the data center;  
(ii) privacy or regulatory concerns with sensitive data; and  
(iii) the need for near-real-time adaptation to dynamic environments (e.g., weather, lighting).
These challenges motivate edge-side, ongoing FM adaptation during missions.
Federated learning (FL) offers a promising solution, enabling collaborative model improvement while avoiding raw data transmission: clients train locally and upload only model parameters to a server for aggregation \cite{chen2024role, yan2025federated}.  
Nevertheless, full-parameter fine-tuning (FPFT) of an entire FM is often impractical on UAVs with limited memory, computing capability, and energy resources.

Low-rank adaptation (LoRA) \cite{hu2022lora}, a prominent parameter-efficient fine-tuning (PEFT) method, greatly reduces the number of trainable parameters (TPs) while achieving accuracy comparable to FPFT.
LoRA is thus attractive for FL in UAV networks because it (i) reduces memory usage by eliminating the need to store optimizer states for frozen parameters and (ii) reduces communication overhead, as UAV clients only need to transmit a small number of TPs for federated aggregation.
However, even with LoRA, hosting and updating an entire FM on resource-constrained UAVs can still be challenging.

Split federated learning (SFL) further alleviates the client-side burden by splitting the FM into a client-side model and a server-side model, while still enabling parallel client training \cite{thapa2022splitfed}. 
However, due to synchronized federated aggregation, the training latency in the SFL is determined by the slowest client, commonly referred to as the ``straggler'' \cite{lin2024split}. 
In UAV networks, the straggler issue is aggravated by the heterogeneous computing and communication capabilities, mobility, and limited energy of UAVs \cite{qiang2025split}.
Split point selection (SPS) significantly influences both training latency and client-side energy consumption, as it determines the computational load on both the client and server, as well as the communication overhead \cite{yan2025federated, lin2024split, qiang2025split, xu2025optimizing, ni2025federated}.
Efficient utilization of server-side computing and communication resources is also critical in mitigating the straggler problem \cite{lin2024split, qiang2025split, xu2025optimizing, ni2025federated}.

\subsection{Related Works and Challenges}
Recently, SPS and communication-computing resource allocation (CCRA) for SFL in wireless networks have received significant attention \cite{xu2023accelerating, chu2025online, qiang2025joint, zhu2024esfl, wen2025training, ao2025semi, wang2025split, hu2025performance, zhang2025efficient, 
zhang2025split, li2025energy, chen2025privacy, 
wu2024joint, yu2025model, wu2025split, 
wu2023split, lin2024efficient, ao2024federated, lin2025adaptsfl, lin2025hasfl, 
qiang2024adaptive,  
chen2025memory, 
zhao2025efficient,
wang2025federated, solat2024split, khan2025qos}.
To mitigate stragglers, existing works explore various strategies, including:
(i) individual SPS for each client \cite{xu2023accelerating, chu2025online, qiang2025joint, zhu2024esfl, wen2025training, ao2025semi, wang2025split, hu2025performance, li2025energy, chen2025privacy, wu2024joint, yu2025model, wu2025split, lin2025adaptsfl, lin2025hasfl, qiang2024adaptive, chen2025memory};
(ii) client grouping \cite{wu2023split, zhang2025efficient};
(iii) parallel split learning without client-side model aggregation \cite{lin2024efficient, ao2024federated}; and
(iv) inter-round asynchronous training \cite{ao2025semi, chu2025online}.
Moreover, SFL for LoRA-based fine-tuning in wireless networks has been studied in \cite{chen2025memory, li2025energy, zhang2025split, chen2025privacy, zhao2025efficient, wang2025federated}, considering aspects such as server energy consumption \cite{li2025energy}, client storage constraints \cite{zhang2025split, chen2025privacy}, privacy preservation \cite{chen2025privacy}, LoRA rank optimization \cite{zhao2025efficient}, and smashed data compression \cite{zhang2025split}.
However, they may incur long per-round latency due to inefficient utilization of downlink communication resources. Specifically, we analyze as follows.
\begin{table}[t]
\centering  
\caption{Comparison of server-side computing and communication resource utilization paradigms}  
\label{table_SFL_paradigm}  %
\begin{tabular}{cccc}  
\toprule 
        Paradigms & \makecell[c]{Server-side\\computing}  & \makecell[c]{Downlink gradient\\transmission}  & Related studies \\ 
\midrule
        SFL-PP & Parallel & Parallel & \cite{xu2023accelerating, chu2025online, zhu2024esfl, wen2025training, ao2025semi, zhang2025efficient, wang2025split, qiang2025joint, hu2025performance, li2025energy, zhang2025split, chen2025privacy, wu2024joint, yu2025model, wu2025split} \\ 
        SFL-SP & Sequential & Parallel & \makecell[c]{\cite{wu2023split, lin2024efficient, ao2024federated, lin2025adaptsfl, lin2025hasfl, zhao2025efficient, qiang2024adaptive, chen2025memory} and\\PipeSFL \cite{gao2024pipesfl}} \\  
        SFL-PS & Parallel & Sequential & CPSFL (ours) \\ 
\bottomrule  
\end{tabular}
\vspace{-0.4cm}
\end{table}

\begin{figure}[t]
\centering
\includegraphics[width=0.99\linewidth]{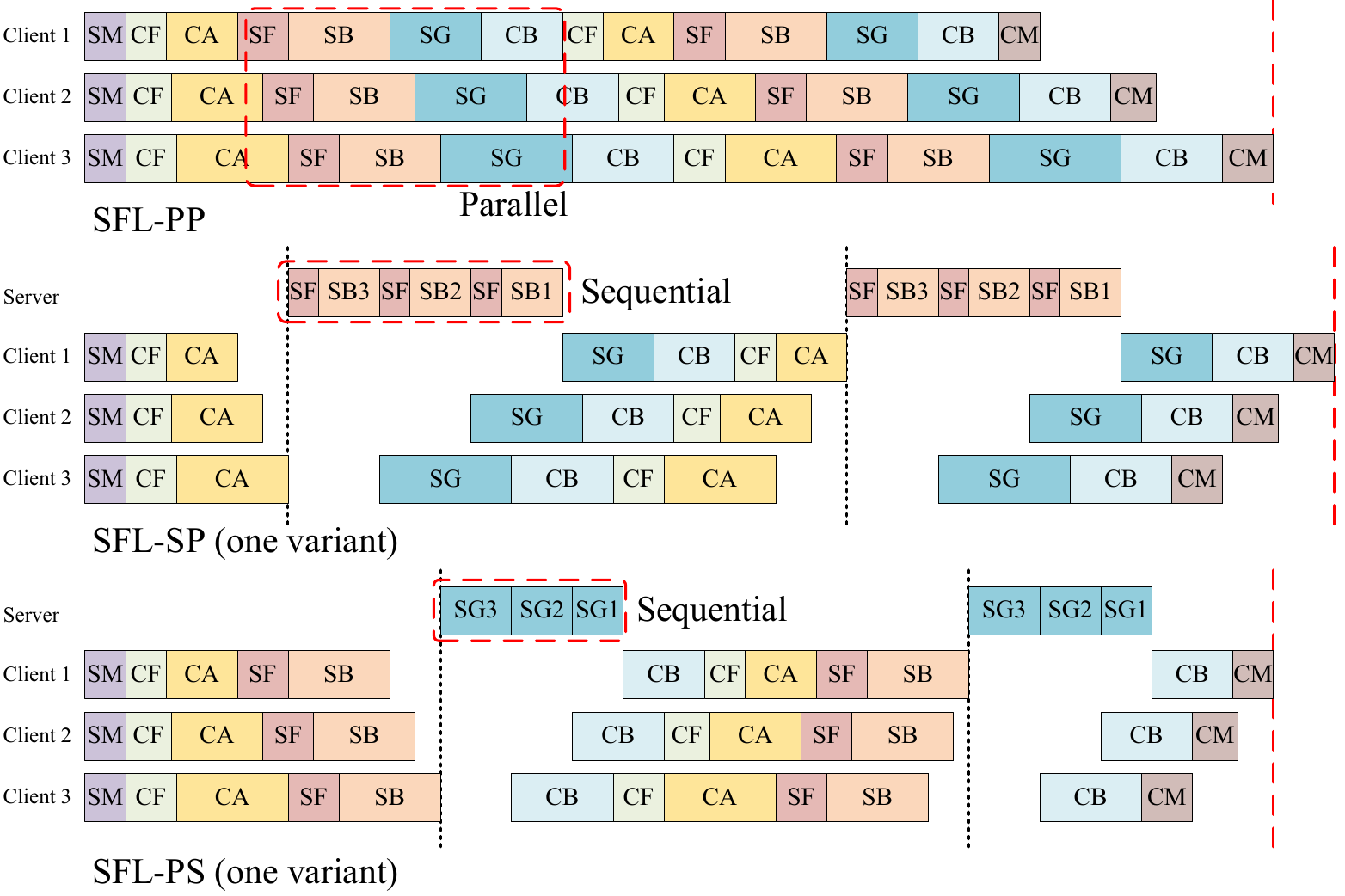}
\caption{The timeline of one training round (consisting of two local iterations) for three clients when SFL-PP (top), SFL-SP (middle), or SFL-PS (bottom) is applied. 
The notations SM, CF, CA, SF, SB, SG, CB, and CM correspond to the eight steps in Fig. \ref{scenario}, respectively.}
\label{Time_diagram_SFL_paradigm}
\vspace{-0.4cm}
\end{figure}

We classify SFL into three types based on whether the server can compute tasks for different clients in parallel and whether it can transmit gradients to different clients simultaneously over the downlink.
Table \ref{table_SFL_paradigm} summarizes these paradigms, referred to as SFL-PP, SFL-SP, and SFL-PS, and their time diagrams are illustrated in Fig. \ref{Time_diagram_SFL_paradigm}.

\textbf{SFL-PP with full parallelism}: 
In SFL-PP, all steps of different clients are executed in parallel.
Specifically, the server allocates a fixed portion of the shared resources, including its computing frequency, the bandwidth for downlink gradient transmission (GT), and corresponding transmit power, to each client and keeps this allocation unchanged in one training round.
SFL-PP is widely adopted in existing works \cite{xu2023accelerating, chu2025online, zhu2024esfl, wen2025training, ao2025semi, zhang2025efficient, wang2025split, qiang2025joint, hu2025performance, li2025energy, zhang2025split, chen2025privacy, wu2024joint, yu2025model, wu2025split} due to its analytical tractability: the per-round training latency equals the maximum per-round latency among all clients.
However, idle resources allocated for one client cannot be utilized by the ongoing computations or downlink GT of other clients, which may increase the per-round latency.
For example, in Fig. \ref{Time_diagram_SFL_paradigm}, when the server begins GT for client 1, the downlink resources allocated to client 3 remain idle.
\textbf{SFL-SP with sequential server-side computing}: 
SFL-SP uses all computing resources for the current computing task and decides the scheduling order of the client tasks.
We refer to the most widely adopted form as vanilla SFL-SP \cite{wu2023split, lin2024efficient, ao2024federated, lin2025adaptsfl, lin2025hasfl, zhao2025efficient, qiang2024adaptive}, which exhibits two synchronizations per local iteration: (i) the server starts computing only after receiving the smashed data from all clients, and (ii) it starts GT only after completing the computing tasks of all clients.
Thus, the vanilla SFL-SP is vulnerable to the straggler issue.  
To mitigate this, the pipelining is introduced to overlap the communication latency and the client computing latency with the server computing latency (the SF and SB steps in Fig. \ref{Time_diagram_SFL_paradigm}) as much as possible \cite{chen2025memory, gao2024pipesfl}.
In \cite{chen2025memory}, the server prioritizes tasks from clients with longer backward propagation (BP) times.
In \cite{gao2024pipesfl}, PipeSFL is proposed to enable finer-grained pipelined computing scheduling via server-side computing priority scheduling and intra-round asynchronous training.
Nevertheless, in wireless networks such as UAV networks, where the computing latency is much smaller than the communication latency, the performance gain offered by PipeSFL may be limited.
\textbf{SFL-PS with sequential downlink GT}: 
Different from SFL-PP and SFL-SP, we propose SFL-PS, which uses all downlink resources for the current GT and decides the scheduling order across clients.
By reducing the GT latency and allowing clients to immediately proceed to subsequent steps upon receiving gradients, SFL-PS has the potential to reduce the per-round latency of SFL in wireless networks.
Notably, SFL-PS has not been explored in existing studies.

In addition to the straggler issue caused by the heterogeneity of UAV clients, their mobility introduces another major challenge.
The SFL with mobile clients, such as vehicles, UAVs, and satellites, has been studied in \cite{qiang2024adaptive, wu2024joint, yu2025model, solat2024split, wu2025split}.
Nevertheless, these works oversimplify the channel variations caused by the moving clients: they either assume the channel remains unchanged within a training round \cite{qiang2024adaptive, solat2024split, wu2025split} or approximate communication rates using time-averaged values over the coverage period \cite{wu2024joint, yu2025model}. 
Given that a single training round often exceeds ten seconds \cite{qiang2024adaptive, yu2025model}, a finer-grained latency analysis is essential.
Building on this, the historical trajectory data can be leveraged to infer future mobility patterns and enable proactive resource management.
To handle the environment with uncertain UAV trajectory and the complexity of the optimization problem, deep reinforcement learning (DRL) techniques are promising for decision-making in SFL scenarios \cite{wu2024joint, yu2025model, qiang2025joint, khan2025qos}.

\subsection{Main Contributions}
In summary, existing SFL training paradigms underutilize downlink communication resources, and existing SFL resource allocation studies inadequately address mobile clients, particularly lacking slot-level channel modeling and trajectory-aware decision-making.  
To tackle these challenges, we propose communication-pipelined SFL (CPSFL) and an attention-based DRL framework for joint SPS and CCRA in CPSFL-enabled UAV networks, supported by the fine-grained latency analysis and the UAV trajectory features extraction.
The main contributions are summarized as follows:
\begin{itemize}
\item To improve downlink resource utilization, we propose SFL-PS, where the server transmits gradients sequentially. 
To mitigate stragglers in vanilla SFL-PS, we further propose CPSFL, incorporating two PipeSFL-inspired enhancements: (i) downlink GT priority scheduling that increases latency overlap by prioritizing clients with larger lags, and (ii) intra-round asynchronous training that reduces downlink idling by enabling immediate GT upon server computation completion.
We also provide the optimality proofs for the scheduling policy under certain simplifying assumptions, along with the latency analysis and comparisons.
\item To finely capture the UAV mobility, we propose a fine-grained latency analysis under SFL-PS where the channel varies per time slot instead of per training round.
Based on this, we formulate an optimization problem to minimize both the per-round training latency and the maximum per-client energy consumption by jointly optimizing SPS and CCRA (i.e., the uplink bandwidth allocation and the server computing frequency allocation).
\item Decision-making at the round start is intractable due to the problem complexity and the unavailability of current-round channel knowledge.
To address this, we design an attention-based DRL framework, where the base station (BS) agent leverages previous round information to determine the split point and the CCRA.
An attention mechanism enables effective feature extraction from variable-length UAV trajectories.
\item Simulation results show that CPSFL achieves lower per-round latency than SFL-PP, PipeSFL, and ablation variants in UAV networks where the communication latency is dominant.
Moreover, the DRL-based CPSFL scheme outperforms its variant that does not leverage UAV trajectory and the fixed CCRA scheme, and approaches the best fixed-SPS scheme.
\end{itemize}

\subsection{Organizations}
The rest of this paper is organized as follows:
Section \ref{secSystemModel} presents the system model, the slot-level latency analysis under the SFL-PS paradigm, and the optimization problem.
Section \ref{secCPSFL} presents the proposed CPSFL.
In Section \ref{secproposedScheme}, we propose the attention-based DRL framework to decide the SPS and the CCRA.
Section \ref{sec_Simulation_Results} shows the simulation results.
Finally, Section \ref{sec_Conclusions} concludes this paper.

Notations:
The lowercase, bold lowercase, and bold uppercase, i.e., $a$, ${\bf{a}}$, and ${\bf{A}}$ are scalar, vector, and matrix, respectively.
$\mathbb{R}^{a \times b}$ denotes the space of $a \times b$ real-valued matrices.
$\left| {\cal A} \right|$ denotes the cardinality of set ${\cal A}$.
$(\cdot )^\top$ denotes transpose.
${\mathbb{E}}\{\cdot\}$ denotes the average operation.

\section{System Model}
\label{secSystemModel}
\begin{figure}[t]
\centering
\includegraphics[width=0.99\linewidth]{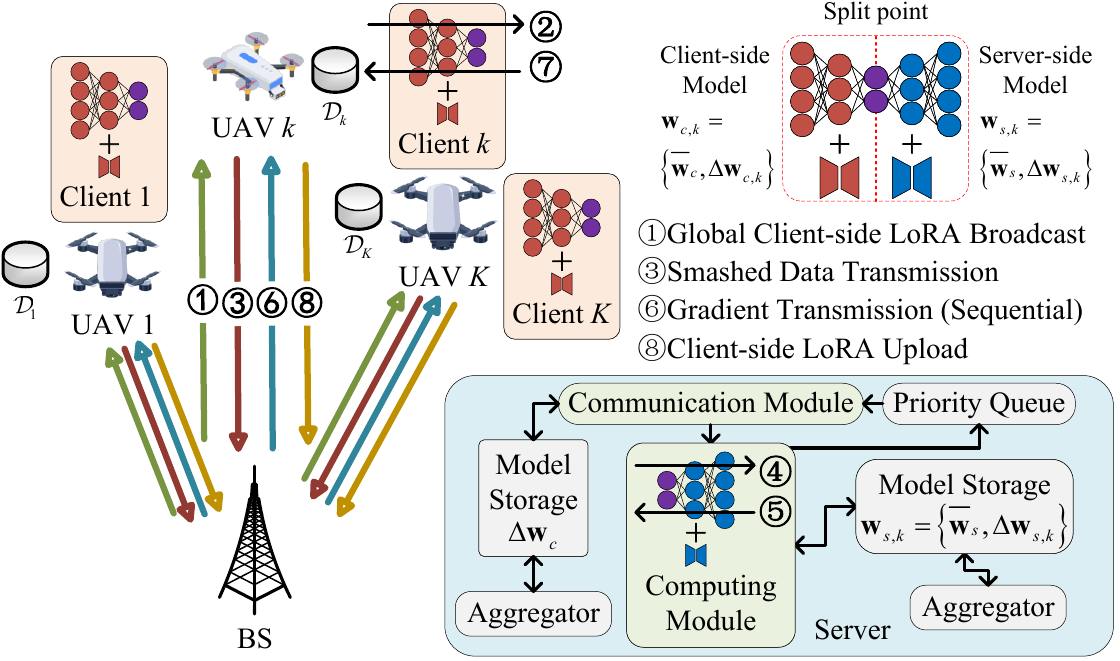}
\caption{Communication-pipelined split federated learning for foundation models LoRA fine-tuning in UAV networks.}
\label{scenario}
\vspace{-0.4cm}
\end{figure}

As shown in Fig. \ref{scenario}, we consider a CPSFL-enabled UAV network, where a BS collaborates with $K$ UAVs to fine-tune a complete FM ${\bf{w}}$ through the CPSFL. 
Specifically, the BS acts as the server, while the UAVs serve as the clients.
The FM is split at the split point $u$ into a server-side FM ${\bf{w}}_s$ and a client-side FM ${\bf{w}}_c$, which are fine-tuned at the BS and the UAVs, respectively. 
In each local iteration, UAV $k$ processes a data batch ${\cal{D}}_k = \left\{{\bf{X}}_k, {\bf{y}}_k\right\}$ of size $B$, where ${\bf{X}}_k = \left\{{\bf{x}}_k^b\right\}_{b=1}^{B}$ is the set of raw data and ${\bf{y}}_k = \left\{y_k^b\right\}_{b=1}^{B}$ is their corresponding labels.
We denote the set of UAV indices as ${\cal K}$.

\subsection{Preliminaries of LoRA and SFL}
\label{secSFL_PipeSFL}

LoRA adapts pre-trained models by injecting trainable low-rank matrices while freezing original weights \cite{hu2022lora}. 
For a weight matrix ${\bf{W}} \in {\mathbb{R}}^{d \times h}$, LoRA represents updates as ${\bf{W}}_0 + \Delta {\bf{W}} = {\bf{W}}_0 + {\bf{B}}{\bf{A}}$, where ${\bf{B}} \in {\mathbb{R}}^{d \times r}$, ${\bf{A}} \in {\mathbb{R}}^{r \times h}$, and the rank $r \ll \min(d, h)$. 
This reduces the number of TPs from $d \times h$ to $(d + h) \times r$ while preserving performance.

We denote the client-side FM as ${\bf{w}}_c = \left\{ {\overline{\bf{w}}}_c, \Delta {\bf{w}}_c \right\}$, where ${\overline{\bf{w}}}_c$ is the frozen pre-trained FM parameters and $\Delta {\bf{w}}_c$ is the client-side TPs, i.e., the LoRA module weights.
Similarly, we denote the server-side FM as ${\bf{w}}_s = \left\{ {\overline{\bf{w}}}_s, \Delta {\bf{w}}_s \right\}$, where $\Delta {\bf{w}}_s$ denotes the TPs including the LoRA module weights and the task module, e.g., the classification head in the classification task.

The SFL process consists of $N$ training rounds, each comprising eight steps as shown in Fig. \ref{scenario}. 
At the start of round $n$, the server broadcasts the latest global client-side TPs $\Delta {\bf{w}}_c(n-1)$ (i.e., $\Delta {\bf{w}}_{c,k}(n, 0)$, $\forall k$) to all clients (Step 1). 
The server set $\Delta {\bf{w}}_{s,k}(n, 0)=\Delta {\bf{w}}_s(n-1)$, $\forall k$.
Subsequently, the $K$ clients perform $I$ local iterations, with each local iteration involving Steps 2 to 7, detailed as follows:
Step 2: Client $k$ performs the forward propagation (FP) of the client-side model ${\bf{w}}_{c,k}(n, i-1)$ and obtains the output (called smashed data) ${\bf{A}}_k(n,i)\!=\!f({\bf{X}}_k(n,i);{\bf{w}}_{c,k}(n, i-\!1))$.
Step 3: Client $k$ transmits ${\bf{A}}_k(n,i)$ and label ${\bf{y}}_k(n,i)$ to the server.
Steps 4 and 5: The server performs the FP and BP of the server-side model ${\bf{w}}_{s,k}(n, i-1)$ and obtains $\Delta {\bf{w}}_{s,k}(n, i)$.
Step 6: The server transmits the gradients of the smashed data ${\bf{G}}_k(n,i)=\nabla \ell \left( {{{\bf{A}}_k}(n,i), {\bf{y}}_k(n,i);{{\bf{w}}_{s,k}(n, i-1)}} \right)$ to client $k$.
Step 7: Client $k$ performs the BP of the client-side model to obtain $\Delta {\bf{w}}_{c,k}(n, i)$.
After completing $I$ local iterations, each client transmits $\Delta {\bf{w}}_{c,k}(n, I)$ to the server (Step 8). 
Finally, the server aggregates the $\Delta {\bf{w}}_{c,k}(n, I)$ from $K$ clients to obtain the updated global client-side TPs $\Delta {\bf{w}}_c(n)$.
Concurrently, the server aggregates the $\Delta {\bf{w}}_{s,k}(n, I)$ to obtain the updated global server-side TPs $\Delta {\bf{w}}_s(n)$, completing one training round.

\subsection{SFL-PS Paradigm}
\label{SFL_PS}
We propose a training paradigm, SFL-PS, characterized by parallel server-side computing and sequential downlink GT, as illustrated in Table \ref{table_SFL_paradigm} and Fig. \ref{Time_diagram_SFL_paradigm}.
Specifically, in steps 4 and 5, the server partitions its computing resources into dedicated portions for individual clients, with the allocation fixed in one training round\footnote{To enable parallel server-side computing, the server model storage should include $K$ server-side FMs ${\bf{w}}_s$ for $K$ clients, respectively.}.
In step 6, all downlink resources are used for the current GT, and the GTs are scheduled in a specific order.
We assume that the uplink bandwidth $W_U$ and the downlink bandwidth $W_D$ do not overlap.
At the start of each round, the server decides the following variables
\begin{itemize}
\item $u\left( n \right)$ is the split point for all clients.
\item ${\alpha _k}\left( n \right)$ is the fraction of server computing frequency allocated to client $k$. 
\item ${\beta _k}\left( n \right)$ is the fraction of uplink bandwidth allocated to client $k$. 
\end{itemize}

Vanilla SFL-PS employs intra-round synchronous training and undesignated GT scheduling.
Specifically, during a local iteration, the server starts GT only after completing the server-side model FP and BP for all clients, and the scheduling order for downlink transmission is undesignated, for example, random or first-come-first-served (FCFS). 
These two characteristics may lead to long per-round training latency.  
To address these limitations, we propose CPSFL in Section \ref{secCPSFL}.

\subsection{Fine-Grained Latency and Energy Consumption Analysis}
In this section, we first analyze the achievable communication rate for each client, and then analyze the time and energy consumption of the eight steps in each training round.
To finely capture the UAV client mobility, we propose a fine-grained latency analysis where the channel varies per time slot instead of per training round.
We denote the channel gain between the BS and UAV $k$ in the $s$-th time slot as $h_k\left( s \right)$, which is a function of the distance between the BS and UAV $k$ and assumed to be constant over a time slot with length $\tau_0$.
At the start of round $n$, the server allocates $W_k\!\left( n \right)={\beta _k}\!\left( n \right) W_U$ bandwidth to client $k$ for the uplink transmission.
The uplink rate from client $k$ to the server in round $n$ is given by
\begin{align}
\label{rate_uplink}
\!{R_{U,k}}\!\left( n,{h_k}\!\left( s \right) \right) \!=\! W_k\!\left( n \right){\log _2}\!\left( 1 \!+ {p_k} {h_k}\!\left( s \right)\!/\! \left(W_k\!\left( n \right)\!{N_0} \right)\right),\!
\end{align}
where $p_k$ is the transmit power of client $k$ and $N_0$ is the noise power spectral density (PSD).
The downlink rate from the server to client $k$ in round $n$ is given by
\begin{align}
\label{rate_downlink_broadcast}
{R_{D,k}}\left( s \right) = W_D{\log _2}\left( 1 + P_S {h_k}\left( s \right)/ \left(W_D{N_0} \right)\right),
\end{align}
where $P_S$ is the total transmit power of the server.
Then, in round $n$, the average communication rate of a uplink transmission step for client $k$, starting from $t_B$ and ending at $t_E$, can be expressed as
\begin{align}
\label{rate_average}
{{\overline R}_{U,k}}&\left( {n,{t_B},{t_E}} \right) = \Big({R_{U,k}}\left( {n,{h_k}\left( {{s_B}} \right)} \right)\left( {\left( {{s_B} + 1} \right){\tau _0} - {t_B}} \right) \nonumber\\
& +\sum_{j = {s_B} + 1}^{{s_E} - 1} {{R_{U,k}}\left( {n,{h_k}\left( j \right)} \right)} {\tau _0} \nonumber\\
& +{R_{U,k}}\left( {n,{h_k}\left( {{s_E}} \right)} \right)\left( {{t_E} - {s_E}{\tau _0}} \right)\Big)/\left( {{t_E} - {t_B}} \right),
\end{align}
where ${s_B} = \left\lfloor {{t_B}/{\tau _0}} \right\rfloor$ and ${s_E} = \left\lfloor {{t_E}/{\tau _0}} \right\rfloor$ are the time slots corresponding to $t_B$ and $t_E$.
If ${s_B} = {s_E}$, then ${{\overline R}_{U,k}}\left( {n,{t_B},{t_E}} \right) = {R_{U,k}}\left( {n,{h_k}\left( {{s_B}} \right)} \right)$.
Similarly, by replacing ${R_{U,k}}\left( n,{h_k}\!\left( s \right) \right)$ with ${R_{D,k}}\left( s \right)$ in \eqref{rate_average}, the average rate for the downlink transmission ${{\overline R}_{D,k}}\left( {t_B},{t_E} \right)$ can be obtained.

\subsubsection{Step 1}
The client-side TPs broadcasting latency is ${\tau _{SM}}\left( n \right) = \max \left\{{\tau _{SM,k}}\left( n \right)\right\}$ and ${\tau _{SM,k}}\left( n \right)$ is given by
\vspace{-0.15cm}
\begin{align}
\label{t_SM}
{\tau _{SM,k}}\left( n \right) = {\Gamma _M}\left( {u\left( n \right)} \right)/{{\overline R}_{D,k}}\left( {{t_B},{t_B} + {\tau _{SM,k}}\left( n \right)} \right),
\end{align}
where $\Gamma _M\left( {u} \right)$ is the data size (in bits) of the client-side TPs when the split point is $u$ and ${t_B} = {t_{SM,k,B}}\left( n \right)$ is the start time of this step.
Note that ${\tau _{SM,k}}\left( n \right)$ is obtained by solving equation \eqref{t_SM}.
The latencies of each communication step below are also obtained by solving the corresponding equations.

\subsubsection{Step 2}
The client-side model FP latency of client $k$ is
\begin{align}
\label{t_CF}
{\tau_{CF,k}}\left( {n,i} \right) = B{\Psi _{CF}}\left( {u\left( n \right)} \right)/\left( {{\kappa _k}{f_k}} \right),
\end{align}
where ${\Psi _{CF}}\!\left( {u} \right)$ is the computation workload (in FLOPs) with one data sample when the split point is $u$, $\kappa _k$ is the computing intensity (in FLOPs/cycle) of client $k$, and $f_k$ is the computing frequency of client $k$. 
The corresponding energy consumption is ${e_{F,k}}\left( {n,i} \right) = {\omega _k}f_k^3{\tau_{CF,k}}\left( {n,i} \right)$,
where $\omega _k$ is the coefficient (in Watt/(cycle/s)$^3$) according to the chip architecture \cite{wu2024joint}.

\subsubsection{Step 3}
The smashed data transmission latency of client $k$ is
\begin{align}
\label{t_CA}
{\tau _{CA,k}}\!\left( {n,i} \right) \!= \!B{\Gamma _A}\!\left( {u\!\left( n \right)} \right)\!/{{\overline R}_{U,k}}\!\left( {n,{t_B},{t_B} \!+\! {\tau _{CA,k}}\!\left( {n,i} \right)} \right)\!,
\end{align}
where $\Gamma _A\left( {u} \right)$ is the data size (in bits) of the smashed data when the split point is $u$ and ${t_B} = {t_{CA,k,B}}\left( {n,i} \right)$.
The corresponding energy consumption is ${e_{A,k}}\left( {n,i} \right) = {p_k}{\tau_{CA,k}}\left( {n,i} \right)$.

\subsubsection{Steps 4 and 5}  
The server-side FP and BP latency for client $k$ is
\begin{align}  
\label{t_SFB}  
\tau_{S,k}\!\left(n,i\right) \!=\! B\!\left({\Psi _{SF}}\!\left( {u\!\left( n \right)} \right) \!+\! {\Psi _{SB}}\!\left( {u\!\left( n \right)} \right)\right) \!/\! \left( {{\kappa _S}{\alpha _k}\!\left( n \right)\!{f_S}} \right),\!
\end{align}  
where $\Psi_{SF}(u)$ and $\Psi_{SB}(u)$ denote the computational workloads (in FLOPs) per data sample for FP and BP, respectively, when the split point is $u$; $\kappa_S$ is the server's computing intensity (in FLOPs/cycle); and $f_S$ is the server's computing frequency.

\subsubsection{Step 6}
The gradient transmission latency for client $k$ is
\begin{align}
\label{t_SG}
{\tau _{SG,k}}\left( {n,i} \right) \!=\! B{\Gamma _G}\!\left( {u\left( n \right)} \right)\!/{{\overline R}_{D,k}}\left( {{t_B},{t_B} + {\tau _{SG,k}}\!\left( {n,i} \right)} \right),\!
\end{align}
where $\Gamma _G\!\left( {u} \right)$ is the data size (in bits) of the gradients of the smashed data when the split point is $u$ and ${t_B} = {t_{SG,k,B}}\left( {n,i} \right)$.

\subsubsection{Step 7}
The client-side model BP latency of client $k$ is
\begin{align}
\label{t_CB}
{\tau_{CB,k}}\left( {n,i} \right) = B{\Psi _{CB}}\left( {u\left( n \right)} \right)/\left( {{\kappa _k}{f_k}} \right),
\end{align}
where ${\Psi _{CB}}\left( {u} \right)$ is the computation workload (in FLOPs) with one data sample when the split point is $u$.
The corresponding energy consumption is ${e_{B,k}}\left( {n,i} \right) = {\omega _k}f_k^3{\tau_{CB,k}}\left( {n,i} \right)$.

\subsubsection{Step 8}
The client-side TPs uplink transmission latency of client $k$ is
\begin{align}
\label{t_CM}
{\tau _{CM,k}}\left( n \right) = {\Gamma _M}\!\left( {u\left( n \right)} \right)/{{\overline R}_{U,k}}\!\left( {n,{t_B},{t_B} + {\tau _{CM,k}}\!\left( n \right)} \right),\!
\end{align}
where ${t_B} = {t_{CM,k,B}}\left( n \right)$.
The corresponding energy consumption is ${e_{M,k}}\left( n \right) = {p_k}{\tau_{CM,k}}\left( n \right)$.

\subsubsection{Total energy consumption for one round}
In round $n$, the total energy consumption of client $k$ is
\begin{align}
\label{energy}
\!\!{e_k}\!\left( n \right) \!= \!\!\sum_{i = 1}^I {\!\left( {{e_{F,k}}\!\left( {n,\!i} \right) \!+\! {e_{A,k}}\!\left( {n,\!i} \right) \!+\! {e_{B,k}}\!\left( {n,\!i} \right)} \right)}  \!+\! {e_{M,k}}\!\left( n \right)\!.\!\!
\end{align}

\subsubsection{Lag of clients}
We define the lag of client $k$ as the time interval from the start of its BP in the previous local iteration to the completion of its server-side computation in the current iteration, as illustrated in Fig. \ref{Time_diagram_CPSFL}.
Specifically, the lag of client $k$ in local iteration $i$ of round $n$ is defined as
\begin{align}
\label{lag_client_k}
l_k\left(n,i\right)=&\tau_{CB,k}\left(n,i-1\right)+\tau_{CF,k}\left(n,i\right)+\nonumber\\
&\tau_{CA,k}\left(n,i\right)+\tau_{S,k}\left(n,i\right).
\end{align}

\subsubsection{The upper bound of the total latency for one round}
The total training latency of round $n$ is denoted by $\tau\left( n \right)$.
Under the SFL-PS paradigm, the upper bound of $\tau\left( n \right)$, denoted as $\tau_{\max}\left( n \right)$, is achieved when the GT of the client with the highest lag in each iteration is performed at the end, i.e.,
\begin{align}
\label{Latency_SFLPS_upper_bound}
\tau\left(n\right)\le&\tau_{\max}\left( n \right)=\max_k\left\{\tau_{SM,k}\left(n\right)\right\}+\nonumber\\
&\max_k\left\{\tau_{CF,k}\left(n,1\right)+\tau_{CA,k}\left(n,1\right)+\tau_{S,k}\left(n,1\right)\right\}+\nonumber\\
&\sum_{i=1}^{I-1}\left(\sum_{k=1}^{K}{\tau_{SG,k}\left(n,i\right)}+\max_k\left\{l_k\left(n,i+1\right)\right\}\right)+\nonumber\\
&\sum_{k=1}^{K}{\tau_{SG,k}\left(n,I\right)}+\max_k\left\{\tau_{CB,k}\left(n,I\right)+\tau_{CM,k}\left(n\right)\right\}.
\end{align}

\subsection{Problem Formulation}
In the SFL-PS-enabled UAV network, the client mobility induces time-varying wireless channels and achievable communication rates, which affect the latency and the energy consumption. Therefore, to minimize the training latency and energy consumption per round, adjusting the SPS and the CCRA when each round begins is necessary.
To prevent UAVs from depleting their energy too quickly, we focus on the maximum energy consumption of UAVs.
The optimization problem can be formulated as 
\begin{subequations}
\label{Problem_ori}
\begin{align}
\label{Problem_ori_obj}
&{\min_{\left\{ {{\alpha _k}\left( n \right),{\beta _k}\left( n \right),u\left( n \right)} \right\}}}\;\tau\left( n \right) +  \lambda {\max}_k \left\{ {{e_k}\left( n \right)} \right\}\\
\label{cons_power_ratio}
&\;\;\;\;\;\;\;\;\;\;\;\;\;\;{\rm{s.t.}}\;\;\;\;\;\;
{\alpha_{\min}} \le {\alpha _k}\left( n \right) \le 1, \sum_{k = 1}^K {{\alpha _k}\left( n \right)}  = 1,\\
\label{cons_bandwidth_ratio}
&\;\;\;\;\;\;\;\;\;\;\;\;\;\;\;\;\;\;\;\;\;\;\;\;\;
{\beta _{\min}} \le {\beta _k}\left( n \right) \le 1, \sum_{k = 1}^K {{\beta _k}\left( n \right)}  = 1,\\
\label{cons_split_point}
&\;\;\;\;\;\;\;\;\;\;\;\;\;\;\;\;\;\;\;\;\;\;\;\;\;
u\left( n \right) \in {\cal U},
\end{align} 
\end{subequations}
where $\lambda>0$ is the weight of the energy term and ${\cal U}$ is the set of possible values of the split point.
Besides, $\alpha_{\min}\in[0, 1/K]$ and $\beta_{\min}\in[0, 1/K]$ are the limits of the transmitting power fraction and the bandwidth fraction, respectively. 
To keep $\tau\left( n \right)$ away from its upper bound $\tau_{\max}\left( n \right)$ in \eqref{Latency_SFLPS_upper_bound}, we propose CPSFL in the next section to enhance vanilla SFL-PS.

\section{Communication-Pipelined SFL}
\label{secCPSFL}
\begin{figure}[t]
\centering
\includegraphics[width=0.99\linewidth]{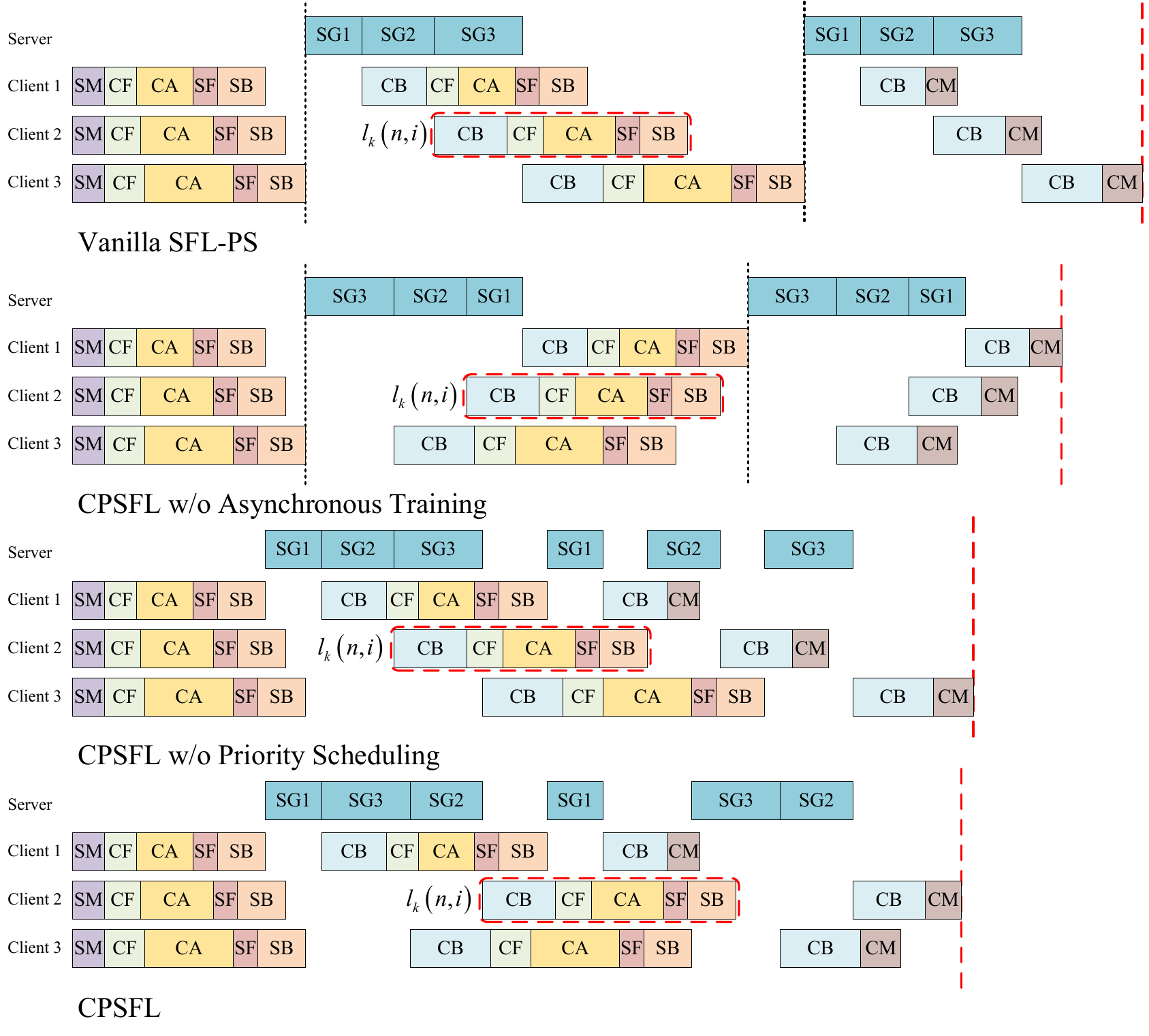}
\caption{The timeline of one training round (consisting of two local iterations) for three clients when the proposed CPSFL and its ablation variants are applied. 
The notations SM, CF, CA, SF, SB, SG$k$, CB, and CM correspond to the eight steps in Fig. \ref{scenario}, respectively.}
\label{Time_diagram_CPSFL}
\vspace{-0.4cm}
\end{figure}
In this section, we propose the CPSFL, which enhances the vanilla SFL-PS through the downlink GT priority scheduling and intra-round asynchronous training.
Specifically, the priority scheduling allows clients who are likely to become stragglers to receive the gradients required for BP earlier. 
The asynchronous training reduces the idle downlink waiting time by enabling immediate GT upon server computation completion.
Fig. \ref{Time_diagram_CPSFL} illustrates the timelines of four paradigms: vanilla SFL-PS, CPSFL without the intra-round asynchronous training (CPSFL w/o AT), CPSFL without the downlink GT priority scheduling (CPSFL w/o PS), and CPSFL, showing that CPSFL achieves the shortest training latency for one round.

In the following, we first present the priority scheduling mechanism in CPSFL w/o AT and establish its optimality via theoretical analysis.
Next, we introduce CPSFL and prove the scheduling optimality under additional simplifying assumptions.
We then analyze the per-round latency of CPSFL and its ablation variants, and finally compare them against the other two paradigms: SFL-PP and PipeSFL.

\subsection{Downlink Gradient Transmission Priority Scheduling}

Due to client heterogeneity in channel conditions and computing capabilities, the order of downlink GT significantly impacts the per-iteration training latency.
Under the synchronous training setting, there are theoretically $K!$ possible scheduling orders with $K$ clients.
Thus, finding the optimal orders is NP-hard, and exhaustive search is prohibitively time-consuming.
To address this, we design a priority scheduling mechanism for downlink GT and prove its optimality.
This paradigm is called CPSFL w/o AT.

As shown in Algorithm \ref{alg_CPSFL_server}, upon completing the computing task for client $k$, the server obtain the gradient ${\bf{G}}_k(n,i)$ for client $k$ and inserts it into the priority queue along with its lag ${l_k}\left( {n,i} \right)$ from \eqref{lag_client_k}, which serves as the transmission priority.
Consequently, clients with larger lags are assigned higher transmission priority and can start BP earlier.
When $I=0$, the term $\tau_{CB,k}\left(n,0\right)$ in \eqref{lag_client_k} can be replaced by $\varsigma{\tau_{CF,k}}\left( {n,1} \right)$ with a constant $\varsigma>0$.

\begin{algorithm}[t] %
\caption{CPSFL (Two Server-Side Improvements)}
\label{alg_CPSFL_server}
\begin{algorithmic}[1]
\STATE \textbf{Procedure}: Downlink Gradient Transmission Priority Scheduling
\IF {the server completes the server-side model FP and BP for client $k$ and obtains the gradient ${\bf{G}}_k(n,i)$}
\STATE The server adds ${\bf{G}}_k(n,i)$ to the priority queue with the lag ${l_k}\left( {n,i} \right)$ in \eqref{lag_client_k} as its priority.
\ENDIF

\STATE \textbf{Procedure}: Intra-Round Asynchronous Training
\IF {the server is not transmitting gradients and the priority queue is not empty}
\STATE The server retrieves the gradient ${\bf{G}}_k(n,i)$ of the highest-priority client $k$ from the priority queue.
\STATE The server transmits ${\bf{G}}_k(n,i)$ to client $k$.
\STATE The server removes ${\bf{G}}_k(n,i)$ from the priority queue. 
\ENDIF
\end{algorithmic}
\end{algorithm}

\subsubsection{Optimality analysis}
\label{Analysis_synchronous}
We first analyze the timeline when the server employs the priority scheduling, based on the following simplifying assumption.
\begin{asp}
\label{assumption_channel_constant_round}
The wireless channel remains constant within each training round.
Thus, we can omit the round index $n$ and local iteration index $i$ for brevity.
Without loss of generality, we reindex the $K$ clients in non-decreasing order of their lags, i.e., $l_1 \leq l_2 \leq \!\cdots\! \leq l_K$, so that their transmission priorities are client 1 $<$ client 2 $<\!\cdots\!<$ client $K$.
\end{asp}

For CPSFL w/o AT, we define the per-iteration training latency as the time interval between two consecutive instants when the server finishes the computing tasks of all clients, which is illustrated in Fig. \ref{Time_diagram_CPSFL} as the gap between two vertical black dashed lines.
This latency is given by
\begin{align}
\label{time_per_iteration}
T_{\rm{iter}}=\max_k{\left\{\sum_{j=k}^{K}{\tau_{SG,j}}+l_k\right\}}.
\end{align}

Then, the following theorem shows the optimality of the proposed GT priority scheduling.
\begin{thm}
\label{optimal_schedule_synchronous}
The optimal strategy to minimize the per-iteration training latency $T_{\rm{iter}}$ is to prioritize the GT task of clients with larger lag, i.e., schedule the GT task of client $m$ before client $k$ if and only if $l_m\geq l_k$, $\forall m,k\in{\cal K}$.
\end{thm}

\begin{proof}
See Appendix \ref{sec_Proof_theorem1}.
\end{proof}

\subsection{Intra-Round Asynchronous Training}
\label{sec_Intra_Round_Asyn_Train}
In CPSFL w/o AT, the server's downlink communication resources remain idle between the completion of the computing task of one client and the completion of the computing tasks of all clients.
To improve resource utilization, the proposed CPSFL adopts the intra-round asynchronous training.  
As shown in Algorithm \ref{alg_CPSFL_server}, whenever the server is not transmitting gradients, it immediately retrieves the highest-priority gradient from the priority queue, starts transmission, and then removes this gradient and its priority from the queue.  
This allows the server to start downlink GT without waiting for the computing tasks of all clients to complete.

Notably, the gradients in the priority queue may originate from different local iteration counts across clients.  
For example, client $j$ may expect to receive a gradient of iteration 2, while client $k$ expects one of iteration 3.  
In this study, the transmission priority of each gradient is independent of the number of local iterations completed by the corresponding client.

\subsubsection{Optimality analysis}
To simplify the per-round latency analysis, we adopt Assumption \ref{assumption_channel_constant_round} and two other assumptions.
\begin{asp}
\label{assumption_latency_CM_negligible}
The uplink transmission latency of the client-side TPs, ${\tau _{CM,k}}$, is negligible.
\end{asp}
Assumption \ref{assumption_latency_CM_negligible} approximately holds in practice for two reasons.  
First, LoRA and model splitting significantly reduce the number of client-side TPs, making the data size $\Gamma_M(u)$ in \eqref{t_CM} much smaller than $B\Gamma_A(u)$ in \eqref{t_CA}, so that $\tau_{CM,k} \ll \tau_{CA,k}$.  
Second, when the number of local iterations $I$ is large, the contribution of $\tau_{CM,k}$ to the total per-round latency becomes negligible.

\begin{asp}
\label{assumption_latency_SG_equal}
The GT latencies for all clients are equal, i.e., $\tau_{SG,k}=\tau_{SG}$, $\forall k$ (based on Assumption \ref{assumption_channel_constant_round}).
\end{asp}
Assumption \ref{assumption_latency_SG_equal} holds when the downlink communication rates between the clients and the server are identical.
Since the server uses all downlink resources (bandwidth and transmit power) for each GT, as shown in \eqref{rate_downlink_broadcast} and \eqref{t_SG}, this assumption is equivalent to assuming identical channel gains: $h_k = h$, $\forall k$.
Moreover, Assumption \ref{assumption_latency_SG_equal} is approximately valid in practice: due to limited client uplink bandwidth and transmit power, $\tau_{SG,k}$ is much smaller than $\tau_{CA,k}$, and the differences among $\tau_{SG,k}$ across clients is relatively small.

Under Assumptions \ref{assumption_channel_constant_round}–\ref{assumption_latency_SG_equal}, we establish the following optimality theorem for the proposed scheduling policy.
\begin{thm}
\label{optimal_schedule_asynchronous}
The optimal strategy to minimize the per-round training latency of CPSFL is to prioritize the GT task of clients with larger lag, i.e., schedule the GT task of client $m$ before client $k$ if and only if $l_m\geq l_k$, $\forall m,k\in {\cal K}$.
\end{thm}

\begin{proof}
See Appendix \ref{sec_Proof_theorem2}.
\end{proof}

\subsection{Per-Round Training Latency Analysis}
\label{latency_analysis}
In the analysis, we omit the broadcast latency of the global client-side TPs ${\tau _{SM}}\left( n \right)$, as it is identical for all clients and paradigms.
Besides, we adopt Assumptions \ref{assumption_channel_constant_round} and \ref{assumption_latency_CM_negligible} for brevity.

\subsubsection{Per-round latency of the vanilla SFL-PS}
We denote the per-round latency of the vanilla SFL-PS by $\tau_1$.
From \eqref{Latency_SFLPS_upper_bound}, we have $\tau_1\le\tau_{\max}$.
The minimum value of $\tau_1$ is achieved when CPSFL w/o AT is applied, i.e., the GT follows the priority scheduling in Theorem \ref{optimal_schedule_synchronous}.
We denote the per-round latency of CPSFL w/o AT by $\tau_2$, so $\tau_1\ge \tau_2$.

\subsubsection{Per-round latency of CPSFL w/o AT}
The per-round latency of CPSFL w/o AT $\tau_2$ is given by 
\begin{align}
\label{Latency_CPSFL_wo_AT}
\tau_2=&\max_k\left\{\tau_{CF,k}+\tau_{CA,k}+\tau_{S,k}\right\}+\left(I-1\right) \cdot \nonumber\\
&\max_k\left\{\sum_{j=k}^{K}\tau_{SG,j}+l_k\right\}+\max_k\left\{\sum_{j=k}^{K}\tau_{SG,j}+\tau_{CB,k}\right\}.
\end{align}
Thus, $\tau_2\geq I\max_k\left\{\sum_{j=k}^{K}\tau_{SG,j}+l_k\right\}={\widehat{\tau}}_2$, where the equality holds when $\arg \max_k\left\{\tau_{CF,k}+\tau_{CA,k}+\tau_{S,k}\right\} = \arg \max_k\left\{\sum_{j=k}^{K}\tau_{SG,j}+\tau_{CB,k}\right\}$ or when $I$ is large.

\subsubsection{Per-round latency of CPSFL}
\label{Latency_CPSFL}
Due to the intra-round asynchronous training, the per-round latency of CPSFL, denoted by $\tau_{\rm{CPSFL}}$, lacks a closed-form expression.  
Instead, we characterize its bounds.
As described in Section \ref{sec_Intra_Round_Asyn_Train}, under asynchronous training, a lower-priority client's GT may start earlier than a higher-priority one's, provided no higher-priority GT task is yet enqueued.
When each client must wait for all higher-priority GTs to complete before its GT begins in every local iteration, $\tau_{\rm{CPSFL}}$ achieves its upper bound, i.e.,
\begin{align}
\label{Latency_CPSFL_max}
\tau_{\rm{CPSFL}}\le \tau_2.
\end{align}
In summary, we have $\tau_{\rm{CPSFL}}\leq \tau_2\leq \tau_1$.

The lower bound is achieved when the downlink GT fully overlaps all other latencies, except for the latency of client 1, the client with the smallest lag, at the start and end of each round.
\begin{align}
\label{Latency_CPSFL_min}
\tau_{\rm{CPSFL}}\!\geq\!\tau_{CF,1}\!+\tau_{CA,1}\!+\tau_{S,1}\!+\sum_{i=1}^{I}\sum_{j=1}^{K}{\tau_{SG,j}}\!+\tau_{CB,1}.\!
\end{align}

\subsection{Comparison of the Per-Round Training Latency of CPSFL, PipeSFL and SFL-PP}
\label{latency_Comparison}
In this section, we compare the approximate upper bound of the per-round latency of CPSFL, ${\widehat{\tau}}_2$, with that of PipeSFL, ${\widehat{\tau}}_3$, and the per-round latency of SFL-PP, $\tau_{\rm{PP}}$, to characterize the conditions under which either CPSFL or PipeSFL outperforms the others.
Firstly, we redefine the latency notation for the relevant steps of SFL-PP and PipeSFL since they utilize the server's computing and communication resources differently from the proposed CPSFL, as shown in Table \ref{table_SFL_paradigm}.  
The downlink rate from the server to client $k$ in round $n$ is given by
\begin{align}
\label{rate_downlink}
{R_{D,k}^\prime}&\left( n,{h_k}\left( s \right) \right) = {\beta _k}\left( n \right) W_D\cdot \nonumber\\
&{\log _2}\left( 1 + {\rho _k}\left( n \right)P_S{h_k}\left( s \right)/ \left({\beta _k}\left( n \right) W_D{N_0} \right)\right),
\end{align}
where ${\rho _k}\left( n \right)$ is the fraction of the server transmit power allocated to the client $k$. 
Its constraint is given by
\begin{align}
\label{cons_power}
{\rho_{\min}} \le {\rho _k}\left( n \right) \le 1, \sum_{k = 1}^K {{\rho _k}\left( n \right)}  = 1.
\end{align}
Then, the average rate for the downlink GT ${{\overline R}_{D,k}^\prime}\left( n,{t_B},{t_E} \right)$ can be obtained by replacing ${R_{U,k}}\left( n,{h_k}\left( s \right) \right)$ with ${R_{D,k}^\prime}\left( n,{h_k}\left( s \right) \right)$ in \eqref{rate_average}.
The GT latency for client $k$ $\tau_{SG,k}^\prime\left( {n,i} \right)$ can be obtained by replacing ${{\overline R}_{D,k}}\left( n,{t_B},{t_E} \right)$ with ${{\overline R}_{D,k}^\prime}\left( n,{t_B},{t_E} \right)$ in \eqref{t_SG}.
To avoid cumbersome and difficult comparisons, we adopt Assumptions \ref{assumption_channel_constant_round} and \ref{assumption_latency_CM_negligible} in the following.
Obviously, we have $\tau_{SG,k}<\tau_{SG,k}^\prime$.

The per-round latency of SFL-PP can be expressed as
\begin{align}
\label{latency_SFLPP_simple}
\tau_{\rm{PP}}=I\max_k\!\left\{\tau_{CF,k}\!+\tau_{CA,k}\!+\tau_{S,k}\!+\tau_{SG,k}^\prime\!+\tau_{CB,k}\right\}.
\end{align}

In PipeSFL, the lag of client $k$ is given by
\begin{align}
\label{lag_client_k_PipeSFL}
l_k^\prime=\tau_{CF,k}+\tau_{CA,k}+\tau_{SG,k}^\prime+\tau_{CB,k}.
\end{align}
Then, we reindex the $K$ clients by their lags, i.e., $l_1^\prime \leq l_2^\prime \leq \cdots \leq l_K^\prime$\footnote{Note that the index $k$ in $\tau_2$ and that in $\tau_3$ may refer to different clients, since $k$ in $\tau_2$ is indexed according to $l_k$ in \eqref{lag_client_k}, whereas $k$ in $\tau_3$ is indexed according to $l_k^\prime$. }.
Moreover, the server-side computing latency of client $k$ is denoted as $\tau_{S,k}^\prime$, which is identical for all clients since the server allocates all its computing resources to the current task and all clients share the same split point. 
For brevity, we denote it as $\tau_S^\prime$.
Obviously, we have $\tau_S^\prime<\tau_{S,k}$.
Similarly to $\tau_2$ in \eqref{Latency_CPSFL_wo_AT} and \eqref{Latency_CPSFL_max}, the upper bound of per-round latency of PipeSFL can be expressed as
\begin{align}
\label{Latency_PipeSFL_wo_AT}
\tau_3=&\max_k\left\{\tau_{CF,k}+\tau_{CA,k}\right\}+\left(I-1\right)\max_k\big\{\left(K-k+1\right)\tau_S^\prime\nonumber\\
&+l_k^\prime\big\}+\max_k\!\left\{\left(K-\!k+\!1\right)\tau_S^\prime+\tau_{SG,k}^\prime+\tau_{CB,k}\right\}.\!
\end{align}
Thus, we have $\tau_3\geq I\max_k\left\{\sum_{j=k}^{K}\tau_S^\prime+l_k^\prime\right\}={\widehat{\tau}}_3$, where the equality holds when $\arg \max_k\left\{\tau_{CF,k}+\tau_{CA,k}\right\} = \arg \max_k\left\{\sum_{j=k}^{K}\tau_S^\prime+\tau_{SG,k}^\prime+\tau_{CB,k}\right\}$ or when $I$ is large.

Next, we compare the latencies under two extreme cases to derive intuitive insights.
\subsubsection{Case 1: negligible computing latency}
When computing latency is negligible, i.e., $\tau_{CF,k}, \tau_{S,k}, \tau_{CB,k}, \tau_S^\prime \rightarrow 0$, we compare the latency of paradigms with parallel GT, $\tau_{\rm{PP}}$ and ${\widehat{\tau}}_3$, against that of the paradigm with sequential GT, ${\widehat{\tau}}_2$.
In this case, ${\widehat{\tau}}_2=I\max_k\left\{\sum_{j=k}^{K}\tau_{SG,j}+\tau_{CA,k}\right\}$ and $\tau_{\rm{PP}}={\widehat{\tau}}_3=I\max_k\left\{\tau_{CA,k}+\tau_{SG,k}^\prime\right\}$.
Thus, if 
\begin{align}
\label{sufficient_negligible_computing}
\max_k\left\{\sum_{j=k}^{K}\tau_{SG,j}+\tau_{CA,k}\right\}\leq\max_k\left\{\tau_{SG,k}^\prime+\tau_{CA,k}\right\},
\end{align}
then ${\widehat{\tau}}_2 \leq {\widehat{\tau}}_3=\tau_{\rm{PP}}$ and the proposed CPSFL is likely to achieve a lower per-round latency.
Next, we present a scenario under which \eqref{sufficient_negligible_computing} holds.
Firstly, \eqref{sufficient_negligible_computing} is equivalent to 
\begin{align}
\label{sufficient_negligible_computing_equal}
\sum_{j=k}^{K}\tau_{SG,j}+\tau_{CA,k}\le\max_k\left\{\tau_{SG,k}^\prime+\tau_{CA,k}\right\}, \forall k.
\end{align}
If the downlink communication resources are equally allocated among clients, i.e., $\beta_k=\rho_k=1/K$, $\forall k$ in \eqref{rate_downlink}, then $K\tau_{SG,k}=\tau_{SG,k}^\prime$, $\forall k$.
Thus, $K\tau_{SG,k}+\tau_{CA,k}\leq \max_k\left\{K\tau_{SG,k}+\tau_{CA,k}\right\}=\max_k\left\{\tau_{SG,k}^\prime+\tau_{CA,k}\right\}$, $\forall k$.
Moreover, if the GT latency ordering is opposite to the lag ordering, i.e., $\tau_{SG,k}\geq\tau_{SG,k+1},\forall k\in\left[1,K-1\right]$ (a relaxed version of Assumption \ref{assumption_latency_SG_equal}), we have $\sum_{j=k}^{K}\tau_{SG,j}\le\left(K-k+1\right)\tau_{SG,k}\le K\tau_{SG,k}$, $\forall k$.
Under these two conditions, \eqref{sufficient_negligible_computing_equal} and \eqref{sufficient_negligible_computing} hold.

\subsubsection{Case 2: negligible communication latency}
When communication latency is negligible, i.e., $\tau_{CA,k},\tau_{SG,k},\tau_{SG,k}^\prime \rightarrow0$, we compare the latency of paradigms with parallel server computing, $\tau_{\rm{PP}}$ and ${\widehat{\tau}}_2$, against that of the paradigm with sequential server computing, ${\widehat{\tau}}_3$.
In this case, ${\widehat{\tau}}_3=I\max_k\left\{\sum_{j=k}^{K}\tau_S^\prime+\tau_{CF,k}+\tau_{CB,k}\right\}$ and $\tau_{\rm{PP}}={\widehat{\tau}}_2=I\max_k\left\{\tau_{CF,k}+\tau_{S,k}+\tau_{CB,k}\right\}$.
Thus, if 
\begin{align}
\label{sufficient_negligible_communication}
\max_k&\left\{\left(K-k+1\right)\tau_S^\prime+\tau_{CF,k}+\tau_{CB,k}\right\}\leq \nonumber\\
&\max_k\left\{\tau_{S,k}+\tau_{CF,k}+\tau_{CB,k}\right\},
\end{align}
then ${\widehat{\tau}}_3\leq {\widehat{\tau}}_2=\tau_{\rm{PP}}$ and the PipeSFL is likely to achieve a lower per-round latency.

Next, we present a scenario under which \eqref{sufficient_negligible_communication} holds.
Firstly, \eqref{sufficient_negligible_communication} is equivalent to $\left(K-k+1\right)\tau_S^\prime+\tau_{CF,k}+\tau_{CB,k}\leq \max_k\left\{\tau_{S,k}+\tau_{CF,k}+\tau_{CB,k}\right\}$, $\forall k$.
If the server computing frequency is equally allocated, i.e., $\alpha _k=1/K$, $\forall k$ in $\tau_{S,k}$ in \eqref{t_SFB}, then $K\tau_S^\prime=\tau_{S,k}$, $\forall k$.
Thus, $K\tau_S^\prime+\tau_{CF,k}+\tau_{CB,k}\le\max_k{\left\{K\tau_S^\prime+\tau_{CF,k}+\tau_{CB,k}\right\}}=\max_k{\left\{\tau_{S,k}+\tau_{CF,k}+\tau_{CB,k}\right\}}$, $\forall k$.
Then, since $\left(K-k+1\right)\tau_S^\prime\le K\tau_S^\prime$, $\forall k$, \eqref{sufficient_negligible_communication} holds.

\section{DRL-based SPS and CCRA for the CPSFL-enabled UAV networks}
\label{secproposedScheme}
In this section, we first motivate the use of DRL for decision-making.
We then formulate the problem as a partially observable Markov decision process (POMDP).  
Next, we describe the design of the BS agent, which incorporates an attention-based mechanism to extract features from historical UAV trajectories.  
Finally, we present the DRL-based SPS and CCRA scheme for the CPSFL-enabled UAV networks.

\subsection{Motivation of DRL-based Solution}
Problem \eqref{Problem_ori} presents three major challenges that render conventional optimization methods impractical:
\begin{enumerate}
\item Complexity: It is a mixed-integer non-linear program (MINLP) with both discrete and continuous variables, making it non-convex and NP-hard \cite{yu2025model}.
\item Analytical intractability: As discussed in Section \ref{Latency_CPSFL}, due to intra-round asynchrony, the per-round latency $\tau_{\text{CPSFL}}(n)$ lacks a closed-form expression in terms of the decision variables and channel strengths.
\item Uncertainty: At the start of each round, the BS has no access to future UAV trajectories or channel realizations.
\end{enumerate}
Furthermore, the problem exhibits temporal correlation: the UAVs' initial locations in the next round depend on the current round's decisions, creating a sequential decision-making process. 
These motivate the use of DRL, a technique naturally suited to optimizing policies in complex, uncertain, and temporally correlated environments.
Due to the environmental uncertainty, the agent only has partial observability.
Thus, we formulate the problem as a POMDP.
\subsection{POMDP Formulation and BS Agent Design}
A POMDP model can be described with a tuple $\left\langle {{\cal S},\Omega, {\cal A},{\cal P}, r,{\cal B},{\gamma}} \right\rangle$.
Specifically, ${\cal S}$ is the state space,
${\bf{o}} \in {\Omega}$ is the observation,
${\bf{a}} \in {\cal A}$ is the action,
$P\left( {{\bf{s}}^{\prime}|{\bf{s}},{\bf{a}}} \right) \in {\cal P}$ is the probabilistic transition function,
$r$ is the reward function, and
$\gamma\in(0,1)$ is the discount factor.
Denote the policy for agent as $\pi:{\Omega}\times{\cal{A}}\to[0,1]$.
The expected discounted cumulative reward for agent is $J\left( {{\pi }} \right) = {{\mathbb{E}}_{{\bf{s}} \left( 0 \right),{{\bf{a}}}\left( 0 \right),...}}\left[ {\sum_{t = 0}^\infty  {\gamma^t{r}\left( t \right)} } \right]$, where ${\bf{a}}\left( t \right) \!\sim\! {\pi }\left( { \cdot |{\bf{o}}\left( t \right)} \right)$ and ${\bf{s}}\left( t\!+\!1 \right)\!\sim\! P\left( { \cdot |{\bf{s}}\left( t \right),{{\bf{a}}\left( t \right)}} \right)$.
The POMDP aims to find an optimal policy ${\pi}^*$ that maximizes $J\left( {{\pi }} \right)$, i.e.,
\begin{align}
\label{POMDP}
{{\pi} ^{*}}= {{\rm{argmax}}}_{\pi} J\left( {\pi} \right).
\end{align}
To find ${\pi}^*$, the DRL algorithms can be applied.

We designate the BS as the agent for deciding the variables listed in Section \ref{SFL_PS}.
Its observation space, action space, and reward function are detailed as follows.

\subsubsection{Attention-based UAV trajectory feature extraction}
To make informed decisions, the BS agent needs to infer future UAV trajectory features and analyze how past trajectories affect the latency and energy consumption.  
Thus, the agent leverages historical trajectory data from the previous training round.
However, the number of time slots per round varies, posing a challenge for fixed-dimension neural networks \cite{sun2025generative}.  
To address this, we employ an attention mechanism with positional encoding to extract fixed-dimensional representations.

Given a query ${\bf{q}} \in {\mathbb{R}}^{D_Q \times 1}$ and $M$ key-value pairs ${\bf{K}}=\left[{\bf{k}}_1,\ldots,{\bf{k}}_M\right]\in{\mathbb{R}}^{D_K\times M}$, ${\bf{V}}=\left[{\bf{v}}_1,\ldots,{\bf{v}}_M\right]\in{\mathbb{R}}^{D_V\times M}$, the attention mechanism produces a feature vector ${\bf{f}}_k \in {\mathbb{R}}^{H \times 1}$.  
Specifically, $\bf{q}$, $\bf{K}$, and $\bf{V}$ are first projected via the learnable weights ${\bf{W}}_Q \in {\mathbb{R}}^{D_S \times D_Q}$, ${\bf{W}}_K \in {\mathbb{R}}^{D_S \times D_K}$, and ${\bf{W}}_V \in {\mathbb{R}}^{H \times D_V}$, respectively. 
To preserve the temporal order, sinusoidal positional encodings are added to the projected keys and values, as well as to the query at its corresponding time step.
The output feature is then obtained via the scaled dot-product attention \cite{vaswani2017attention}
\begin{align}
\label{Attention}
{\bf{f}} = &\left({\bf{W}}_V{\bf{V}} + {\bf{P}}_V\right) \cdot \nonumber\\
&{\rm{softmax}}\left( \frac{1}{\sqrt{D_S}} \!\left( {\bf{W}}_K{\bf{K}} + {\bf{P}}_K \right)^\top \left( {\bf{W}}_Q{\bf{q}} + {\bf{p}}_Q \right) \right),\!
\end{align}
where ${\bf{P}}_V \in {\mathbb{R}}^{H \times M}$, ${\bf{P}}_K \in {\mathbb{R}}^{D_S \times M}$, and ${\bf{p}}_Q \in {\mathbb{R}}^{D_S \times 1}$ are the positional encodings.

For UAV $k$ at time slot $s$, we form a vector ${\boldsymbol{\chi}}_k(s) \in {\mathbb{R}}^{4 \times 1}$ by combining its 3D coordinates and distance to the BS, which is available to the BS at the start of each slot. 
Let ${\cal{S}}\left( n \right)$ denote the set of time slots in round $n$, with $M\left( n \right)=\left| {\cal{S}}\left( n \right)\right|$ varying across rounds.
The query is set to ${\bf{q}} = {\boldsymbol{\chi}}_k(s^{\prime})$, where $s^\prime$ is the last time slot in ${\cal{S}}\left(n-1 \right)$.
The key-value pairs are defined as ${\bf{k}}_m = {\bf{v}}_m = {\boldsymbol{\chi}}_k(s)$, $\forall s \in {\cal{S}}\left( {n - 1} \right)$.
Applying \eqref{Attention} yields a fixed-dimensional trajectory feature ${\bf{f}}_k\left( {n - 1} \right) \in {\mathbb{R}}^{H \times 1}$ for UAV $k$ over round $n-1$.

\subsubsection{Observation space}
The observation of the BS is designed as
\begin{align}
\label{State_Space_server}
{\bf{o}}(n) = \big[ &u\left( {n - 1} \right),\left[{\alpha _k}\left( n-1 \right), {\beta _k}\left( {n - 1} \right),{e_k}\left( {n - 1} \right)\right]_{k\in {\cal K}}, \nonumber\\
&\tau\left( n-1 \right), \left[{\boldsymbol{\chi}}_k(s)\right]_{k\in {\cal K}, s \in {\cal{S}}\left( {n - 1} \right)} \big].
\end{align}
To obtain ${\bf{o}}(n)$, the communication overhead is $K$ positive real numbers (PRNs) (i.e., $e_k\left( {n - 1} \right)$, $\forall k$), which is negligible compared with $\Gamma_M$ and $\Gamma_A$.
After obtaining ${\bf{f}}_k\left( {n - 1} \right)$, $\forall k$, we concatenate it with the other components of ${\bf{o}}(n)$ as input to the subsequent multi-layer perception (MLP).

\subsubsection{Action space}
The action space of the BS is designed as
\begin{align}
\label{Action_Space_server}
{\bf{a}}(n)= \left[ u(n), \widetilde{\boldsymbol{\alpha}}(n), \widetilde{\boldsymbol{\beta}}(n)\right],
\end{align}
where $\widetilde{\boldsymbol{\alpha}}=\left[{\widetilde{\alpha}}_1, \cdots,{\widetilde{\alpha}}_K\right]$, $\sum_{k= 1}^{K}{{{ \tilde{\alpha}}_k}\left(n\right)}= 1$, and ${{ \tilde{\alpha}}_k}\left(n\right)>0$.
$\widetilde{\boldsymbol{\beta}}$ is defined analogously.
To satisfy \eqref{cons_power_ratio}, we obtain ${\alpha _k}\left( n \right)$ by the following linear scaling
\begin{align}
\label{Action_trans}
{\boldsymbol{\alpha}}(n)= 
\begin{cases} 
    \tilde{\boldsymbol{\alpha}}(n),\min\left(\tilde{\boldsymbol{\alpha}}(n)\right)\geq\alpha_{\min}, \\ 
    A\left(\tilde{\boldsymbol{\alpha}}(n)-\min\left(\tilde{\boldsymbol{\alpha}}(n)\right)\right)+\alpha_{\min}, {\rm{else}},
\end{cases}
\end{align}
where $A=(1-K\alpha_{\rm{min}})/\left(1-K\min\left(\widetilde{\boldsymbol{\alpha}}(n)\right)\right) \ge 0$. Besides, ${\beta _k}\left( n \right)$ can be obtained in a similar way to satisfy \eqref{cons_bandwidth_ratio}.

\subsubsection{Reward function}
The reward function of the BS is designed as
\begin{align}
r(n) = -\tau\left( n \right) - \lambda  {\max}_k \left\{ {{e_k}\left( n \right)} \right\}.
\end{align}

\subsection{The Overall DRL Scheme}
\addtolength{\topmargin}{+0.05in} %

\begin{algorithm}[t] %
\caption{DRL-based SPS and CCRA scheme for the CPSFL-enabled UAV networks}
\label{alg_DRL_SPS_CCRA_SFL_UAV}
\small
\begin{algorithmic}[1]
\STATE Initialize the whole model ${\bf{w}}$ with the LoRA modules and the batch size $B$.
\STATE {The BS initializes the policy network $\pi_{\boldsymbol{\vartheta}}$, the value network $V_{\boldsymbol{\varphi}}$, and the trajectory collector $\xi$ of size $B_m$.}
\FOR {round $n=0, \dots, N$} 
\STATE If $n\geq 1$, the BS obtains observation ${\bf{o}}(n)$. 
\STATE If $n\geq 2$, the BS stores the experience $\langle {\bf{o}}(n-1), {\bf{a}}(n-1), {r}(n-1), {\bf{o}}(n) \rangle$ into $\xi$.
\IF {$\xi$ is full}
\STATE The BS obtains all experiences $\langle {{{\bf{o}}_{j}},{{\bf{a}}_{j}},{r_{j}},{{\bf{o}}_{j + 1}}} \rangle$, $j = 1, \ldots, B_m$ in $\xi$ and updates $\pi_{\boldsymbol{\vartheta}}$ and $V_{\boldsymbol{\varphi}}$ by the PPO algorithm, and then clear $\xi$.
\ENDIF
\STATE The BS obtains the action ${\bf{a}}(n)$, decides the split point $u(n)$, the server computing frequency allocation ${\alpha _k}\left( n \right)$, $\forall k$, and the uplink bandwidth allocation ${\beta _k}\left( n \right)$, $\forall k$.
\STATE The BS splits ${\bf{w}}(n-1)$ into ${\bf{w}}_c(n-1)$ and ${\bf{w}}_s(n-1)$, and sends $\Delta {\bf{w}}_c(n-1)$, $u(n)$, and ${\beta _k}\left( n \right)$ to UAV $k$, $\forall k$.

\STATE The UAVs and BS perform $I$ local iterations by the CPSFL paradigm in Algorithm \ref{alg_CPSFL_server}.

\STATE After receiving $\Delta {\bf{w}}_{c,k}(n,I)$, $\forall k$, the BS calculates their weighted average to obtain $\Delta {\bf{w}}_{c}(n)$.
Concurrently, the BS aggregates the $\Delta {\bf{w}}_{s,k}(n, I)$, $\forall k$ to obtain $\Delta {\bf{w}}_s(n)$.

\STATE UAV $k$ send $e_k\left( n \right)$ to the BS, $\forall k$.
\STATE The BS receives reward $r(n)$.
\ENDFOR

\RETURN Trained whole TPs consisting of $\Delta {\bf{w}}_c$ and $\Delta {\bf{w}}_s$.
\end{algorithmic}
\end{algorithm}
The proposed DRL-based SPS and CCRA scheme for the CPSFL-enabled UAV networks is summarized in Algorithm \ref{alg_DRL_SPS_CCRA_SFL_UAV}.
The BS agent is equipped with the proximal policy optimization (PPO) algorithm \cite{schulman2017proximal}.

\section{Simulation Results}
\label{sec_Simulation_Results}
In this section, we present the simulation results to demonstrate the performance improvements brought by the proposed CPSFL in Algorithm \ref{alg_CPSFL_server} for the UAV network and the proposed DRL-based SPS and CCRA scheme in Algorithm \ref{alg_DRL_SPS_CCRA_SFL_UAV}.
\subsection{Benchmarks}
\label{simulation_benchmarks}
We compare CPSFL against the following SFL paradigms and ablation variants:
\begin{itemize}
\item CPSFL w/o AT: CPSFL with intra-round synchronous training (see Fig. \ref{Time_diagram_CPSFL}): the server starts GT only after completing server-side computing for all clients.
\item CPSFL w/o PS: CPSFL with FCFS scheduling instead of priority scheduling.
\item PipeSFL \cite{gao2024pipesfl}: A state-of-the-art SFL-SP variant featuring server-side computing priority scheduling and intra-round asynchronous training\footnote{For fair comparison, we adopt PipeSFL-V2, where each client's server-side FP is immediately followed by its BP.}.
\item PipeSFL w/o AT: PipeSFL with intra-round synchronous training: the server starts computing only after receiving smashed data from all clients.
\item PipeSFL w/o PS: PipeSFL with FCFS scheduling instead of priority scheduling.
\item SFL-PP: All SFL steps of different clients is executed in parallel (see Fig. \ref{Time_diagram_SFL_paradigm}).
\end{itemize}
The latency of SFL-PP can be readily extended from \eqref{latency_SFLPP_simple}.  
In PipeSFL, upon receiving the smashed data, the server calculates the client lag using an extension of \eqref{lag_client_k_PipeSFL}, where we set ${\tau_{CB,k}}\left( {n,0} \right) = \varsigma{\tau_{CF,k}}\left( {n,1} \right)$ and estimate the GT latency by utilizing the current-slot channel information, i.e., ${\tau_{SG,k}^\prime}\left( {n,0} \right) = B{\Gamma _G}\left( {u\left( n \right)} \right)/ {R_{D,k}^\prime}\left( n,{h_k}\left( s_{CA,k,E}\left( {n,i} \right)  \right) \right)$, with $s_{CA,k,E}\left( {n,i} \right) = \left\lfloor {\left({t_{CA,k,B}}\left( {n,i} \right)+ {\tau _{CA,k}}\left( {n,i} \right)\right)/{\tau _0}} \right\rfloor$\footnote{In \cite{gao2024pipesfl}, these terms are set to zero, which may result in priority assignments that deviate from the optimal scheduling in Theorem 2 of \cite{gao2024pipesfl} for minimizing the per-round training latency.}.

\subsection{Simulation Parameters Setting}
We consider a UAV network with a BS and $K=10$ UAVs. 
The BS is located at [0,0,30]${\rm{m}}$.
The horizontal locations of the UAVs are within three concentric annular regions centered on the BS.
UAVs 1-3, 4-6, and 7-10 are located in the inner, middle, and outer rings, respectively.
The boundaries of these regions are 100${\rm{m}}$, 550${\rm{m}}$, 820${\rm{m}}$, and 1000${\rm{m}}$ away from the BS.
The UAVs' height is 20${\rm{m}}$, and their speed is maintained between 0.1${\rm{m/s}}$ and 4${\rm{m/s}}$.
The time slot length is set to $\tau_0=0.1{\rm{s}}$.
The center frequency is $f_c=2{\rm{GHz}}$.
The uplink bandwidth and the downlink bandwidth are $W_U=W_D=20{\rm{MHz}}$.
The noise PSD is set to $N_0=-114{\rm{dBm/MHz}}$.
The path loss $L_{{\rm{RMa}},k}$ in dB follows the RMa-AV LOS model in 3GPP TR 36.777 \cite{3GPP36777}.
The channel gain is $h_k(s)=10^{-L_{{\rm{RMa}},k}(s)/10}$.

\begin{table}[t]
\centering  
\caption{Amount of the FP computation, the client-side TPs, and the smashed data of Swin-L with LoRA}  
\label{table_FLOPs}  %
    \begin{threeparttable}
\begin{tabular}{cccccc}  
\toprule 
$u$ & \makecell[c]{$\Psi_{CF}$\\(GFLOPs)\tnote{1}}  & \makecell[c]{$\Psi_{SF}$\\(GFLOPs)\tnote{1}}  & \makecell[c]{$\Gamma_M$\\(KB)\tnote{2}}  & \makecell[c]{$\Gamma_A$\\(KB)\tnote{2}}  & \makecell[c]{Smashed\\Data Shape}  \\   
\midrule
1 & 6.18 & 64.10 & 192 & 2352 & [56, 56, 192] \\
2 & 12.50 & 57.78 & 612 & 1176 & [28, 28, 384] \\
3 & 64.19 & 6.09 & 7596 & 588 & [14, 14, 768] \\
4 & 70.28 & 0.001 & 9276 & 294 & [7, 7, 1536] \\
\bottomrule  %
\end{tabular}
    \begin{tablenotes}    %
        \footnotesize               %
        \item[1] It is calculated by fvcore with 1Mult-Adds $\approx$ 2FLOPs.
        \item[2] It is measured in float32 precision.
    \end{tablenotes}            %
    \end{threeparttable}
\vspace{-0.4cm}
\end{table}
The BS and UAVs collaborate to fine-tune a Swin-L model ($\approx$200M parameters)\footnote{CPSFL is also applicable to billion-scale models, e.g., SwinV2-G (3.0B).} \cite{liu2021swin} on the CIFAR-100 dataset, where the input size is 3*224*224 and $B=8$.
The LoRA modules of rank 8 are injected into all linear layers except the classification head.
All parameters except the LoRA modules and the classification head are frozen.
Split point $u$ allocates $u$ stages to the client and the remainder to the server.
Thus, ${\cal {U}}=\left\{ {1, \ldots , 4} \right\}$ in \eqref{cons_split_point}.
The values of ${\Psi _{CF}}\left( {u} \right)$, ${\Psi _{SF}}\left( {u} \right)$, $\Gamma _M\left( {u} \right)$, and $\Gamma _A\left( {u} \right)$ are shown in Table \ref{table_FLOPs}.
Besides, we set $\Gamma _G\!\left( {u} \right)=\Gamma _A\!\left( {u} \right)$, ${\Psi _{SB}}\!\left( {u} \right)=\varsigma{\Psi _{SF}}\!\left( {u} \right)$, ${\Psi _{CB}}\!\left( {u} \right)=\varsigma{\Psi _{CF}}\!\left( {u} \right)$, and $\varsigma=2$ since the computations required for BP are about twice that of FP \cite{qiang2024adaptive}.
The BS is equipped with a GeForce RTX 4080 for server-side computing, which provides a computing capacity of 780 AI TOPS ($\approx$195TFLOPS) and operates at a frequency of $f_S = 2.51$GHz \cite{nvidia2025gpucompare}. 
Accordingly, we set its computational intensity to $\kappa_S = 195 \times 10^{12} / f_S$ FLOPs/cycle.
In addition, we set $P_S=40{\rm{W}}$ and $\alpha_{\min} = \beta_{\min} =1/K/5$ in \eqref{cons_power_ratio} and \eqref{cons_bandwidth_ratio}.
The UAVs have heterogeneous transmit powers, with $p_k$ for UAV 1-10 set to 1, 0.4, 0.1, 1, 0.4, 0.1, 1, 0.4, 0.2, and 0.1 W, respectively.
Each UAV is equipped with a device whose performance is comparable to the Jetson Xavier NX for client-side computation, offering a computing capacity of 10TOPS ($\approx$2.5TFLOPS) and operating at a frequency of $f_k = 1$GHz \cite{nvidia2025jetsonxavier}.  
Accordingly, we set its computational density to $\kappa_k = 2.5 \times 10^{12} / f_k$ FLOPs/cycle and set the energy consumption coefficient to $\omega _k=16$W/(GHz)$^3$.
The weight of the energy term is set as $\lambda=4$.

For the BS agent, we set $D_S=8$ and $H=16$ in \eqref{Attention}.
Both $\pi_{\boldsymbol{\vartheta}}$ and $V_{\boldsymbol{\varphi}}$ have three hidden layers with 128, 64, and 32 neurons, respectively. 
Besides, the third hidden layer and the output layer of $\pi_{\boldsymbol{\vartheta}}$ include three branches for deciding ${\alpha _k}\left( n \right)$, ${\beta _k}\left( n \right)$, and $u\left( n \right)$, respectively. The softmax activation function is used in the outputs of the branches for $u\left( n \right)$ to normalize the probabilities.
The parameters of Algorithm \ref{alg_DRL_SPS_CCRA_SFL_UAV} are set as follows: 
discount factor is $\gamma=0.5$, 
learning rates are $\alpha_V=\alpha_{\pi}=3\times 10^{-4}$, and we set $B_m=12$.

\subsection{Performance Evaluation with Fixed Variables}
\label{simulation_fix_variable}
In this section, we perform comparisons under the following fixed resource allocation settings.
For CPSFL, we set $\alpha _k = \beta_k=1/K$, $\forall k$.
For PipeSFL, we set $\rho _k = \beta_k=1/K$, $\forall k$.
For CPSFL, we set $\alpha _k = \beta_k = \rho _k =1/K$, $\forall k$.
The split point is set as $u=2$ by default.
The number of local iteration is set as $I=3$ by default.

\begin{figure}[t]
\centering
\includegraphics[width=0.80\linewidth]{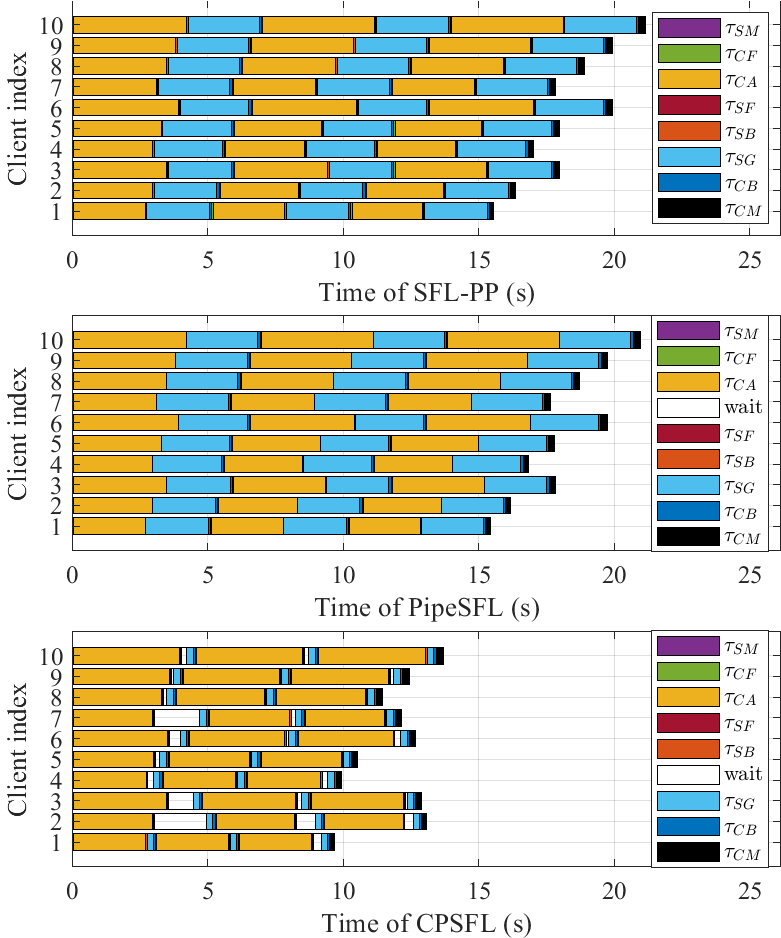}
\caption{Timeline of one training round (consisting of two local iterations) for ten clients when SFL-PP (top), PipeSFL (middle), or CPSFL (bottom) is applied.}
\label{Timeline_SFL_paradigm}
\vspace{-0.4cm}
\end{figure}

Fig. \ref{Timeline_SFL_paradigm} compares the timelines of SFL-PP, PipeSFL, and CPSFL, intuitively illustrating the source of CPSFL's performance gain.
In our simulation setting, the per-round training latency is dominated by the wireless communication, while the client and server computing latencies are relatively small. 
Consequently, although PipeSFL dedicates all server computing resources on the current task and incorporates priority scheduling and asynchronous training, its performance improvement is limited. 
In contrast, CPSFL, an advanced SFL-PS variant, focuses all downlink communication resources on the current GT, significantly reducing the GT latency for each client.
Moreover, CPSFL integrates priority scheduling and asynchronous training, further lowering overall latency.

\begin{figure}[t]
\centering
\includegraphics[width=0.7\linewidth]{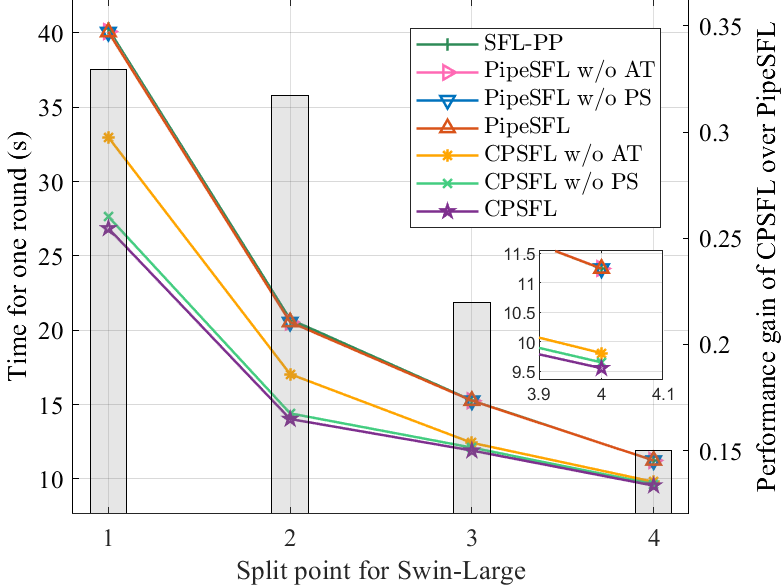}
\caption{Average per-round latency with different split points $u$ obtained by CPSFL and the benchmarks.}
\label{latency_round_fix_variable_split_points}
\vspace{-0.4cm}
\end{figure}

Fig. \ref{latency_round_fix_variable_split_points} shows the average per-round latency ${\overline \tau}\left( n \right)$ over 500 training rounds for different split points $u$.
First, increasing $u$ leads to a decrease in ${\overline \tau}\left( n \right)$ across all schemes since it reduces $\Gamma_A\left( u \right)$ (see Table \ref{table_FLOPs}), thereby decreasing $\tau_{CA,k}$, $\tau_{SG,k}$, and $\tau_{SG,k}^\prime$.
Moreover, CPSFL achieves the lowest latency. 
When $u = 2$, it reduces ${\overline \tau}\left( n \right)$ by over 30\% compared to PipeSFL.  
This gain is more pronounced for smaller $u$ since in this case, the communication latency constitutes a larger fraction of $\tau\left( n \right)$, and the benefit of time-division sequential GT becomes significant.

\begin{table}[t]
\centering  
\caption{Transmit power of five UAV clusters ($K$=12)}
\label{table_UAV_transmit_power}  %
\setlength{\tabcolsep}{3pt} %
    \begin{tabular}{cccc}
    \toprule 
        Clusters & $p_k$ (W) (inner ring) & $p_k$ (W) (middle ring) & $p_k$ (W) (outer ring) \\ 
        \midrule
        1 & 0.2, 0.2, 0.2, 0.2 & 0.4, 0.4, 0.4, 0.4 & 0.8, 0.8, 0.8, 0.8 \\ 
        2 & 0.4, 0.4, 0.4, 0.4 & 0.8, 0.8, 0.2, 0.2 & 0.4, 0.4, 0.4, 0.4 \\ 
        3 & 0.4, 0.4, 0.4, 0.4 & 0.8, 0.8, 0.8, 0.8 & 1.6, 0.2, 0.2, 0.2 \\ 
        4 & 0.8, 0.8, 0.8, 0.4 & 0.8, 0.8, 0.2, 0.2 & 0.4, 0.2, 0.2, 0.2 \\ 
        5 & 0.8, 0.8, 0.8, 0.4 & 0.8, 0.8, 0.4, 0.4 & 0.8, 0.1, 0.1, 0.1 \\ 
    \bottomrule
\end{tabular}
\vspace{-0.3cm}
\end{table}

\begin{figure}[t]
\centering
\includegraphics[width=0.7\linewidth]{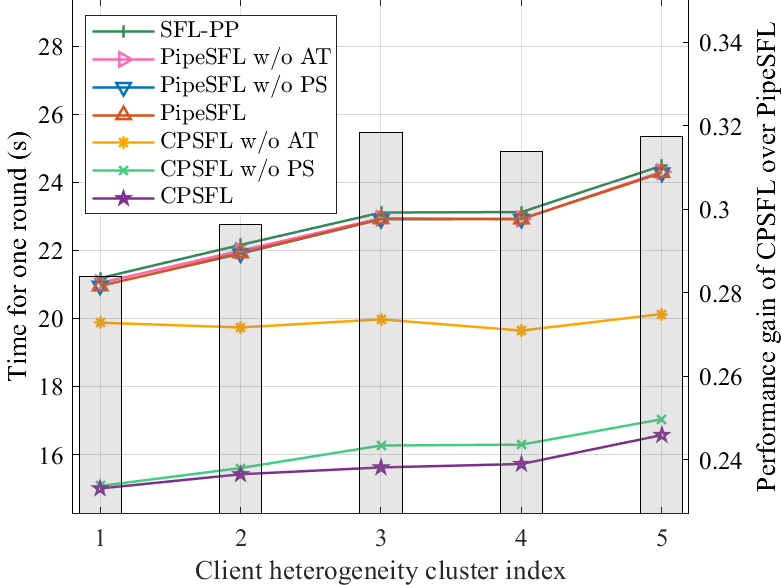}
\caption{Average per-round latency with different degrees of client heterogeneity obtained by CPSFL and the benchmarks.}
\label{latency_round_fix_variable_hetero}
\vspace{-0.4cm}
\end{figure}
In Table \ref{table_UAV_transmit_power}, we define five clusters of UAVs with approximately increasing heterogeneity in their uplink communication latencies $\tau_{CA,k}$, achieved by changing their transmit powers $p_k$ while keeping their average uplink communication rates $R_{U,k}$ nearly the same.
In cluster 1, $\tau_{CA,k}$ is nearly identical across clients, whereas in cluster 5, the last two UAVs exhibit $\tau_{CA,k}$ values nearly 1.5 times larger than those of the others.
Fig. \ref{latency_round_fix_variable_hetero} shows the average per-round latency over 500 training rounds for these clusters.
CPSFL achieves the lowest latency, and its performance gain over CPSFL w/o PS and PipeSFL basically grows with increasing latency heterogeneity.  
This is because CPSFL effectively hides the impact of clients with large lag $l_k$ in the GT latencies $\tau_{SG,k}$ via priority scheduling and intra-round asynchronous training.
In contrast, PipeSFL struggles to hide the latency of clients with large lag $l_k^\prime$ in the server computing latencies $\tau_S^\prime$ since $\tau_S^\prime$ is small in our wireless UAV setting.

\begin{figure}[t]
\centering
\includegraphics[width=0.7\linewidth]{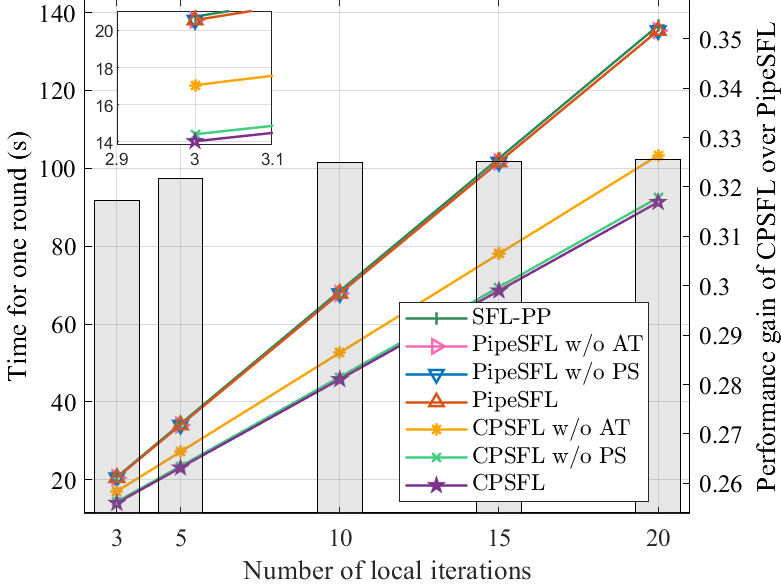}
\caption{Average per-round latency with different numbers of local iterations $I$ obtained by CPSFL and the benchmarks.}
\label{latency_round_fix_variable_iterations}
\vspace{-0.4cm}
\end{figure}

\begin{figure}[t]
\centering
\includegraphics[width=0.7\linewidth]{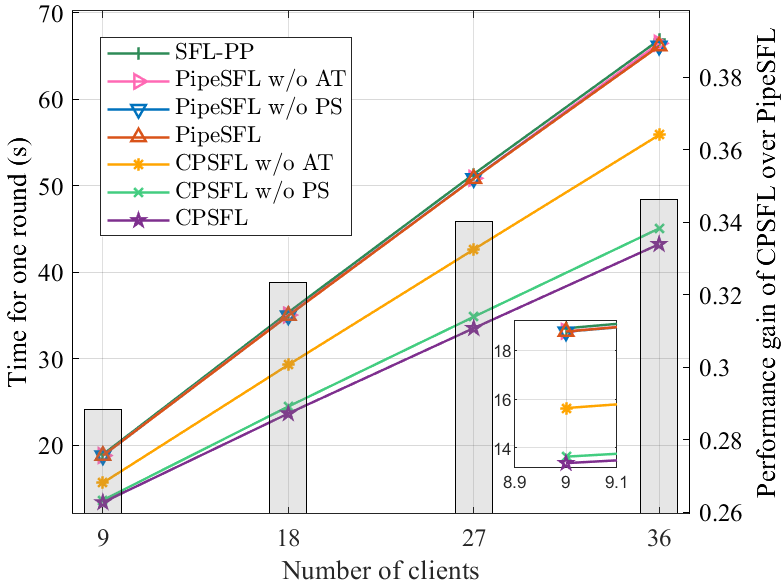}
\caption{Average per-round latency with different numbers of clients $K$ obtained by CPSFL and the benchmarks.}
\label{time_round_fix_variable_num_client}
\vspace{-0.4cm}
\end{figure}
Fig. \ref{latency_round_fix_variable_iterations} and Fig. \ref{time_round_fix_variable_num_client} show the average per-round latency ${\overline \tau}\left( n \right)$ over 500 training rounds with different numbers of local iterations $I$ and clients $K$, respectively.
For Fig. \ref{time_round_fix_variable_num_client}, there are $K/3$ clients in each of the three rings, with transmit powers $p_k$ set to 0.9W (inner), 0.3W (middle), and 0.1W (outer), respectively.  
In both figures, CPSFL achieves the lowest latency across all settings, with the performance gain over PipeSFL slightly increases with larger $I$ or $K$.
Specifically, in Fig. \ref{latency_round_fix_variable_iterations}, ${\overline \tau}\left( n \right)$ grows approximately linearly with $I$.
In Fig. \ref{time_round_fix_variable_num_client}, increasing $K$ raises ${\overline \tau}\left( n \right)$ since the total resource, including the bandwidth and server computing frequency, is fixed.

\vspace{-0.1cm}
\subsection{Performance Evaluation with Optimized Variables}
In this section, we compare the DRL-based scheme in Algorithm \ref{alg_DRL_SPS_CCRA_SFL_UAV} and the fixed variants.
The schemes are appended with the suffixes to indicate how the variables are determined.
\begin{itemize}
\item DRL (proposed): The SPS and CCRA are decided by Algorithm \ref{alg_DRL_SPS_CCRA_SFL_UAV}.
Besides, DRL(-) denotes a variant that replaces $\left[{\boldsymbol{\chi}}_k(s)\right]_{s \in {\cal{S}}\left( {n - 1} \right)}$, $\forall k$ in \eqref{State_Space_server} with the distances from the BS to all UAVs at the last time slot in ${\cal{S}}\left(n-1 \right)$.
Thus, DRL(-) does not involve the attention mechanism.
\item fixed RA: The SPS is decided by Algorithm \ref{alg_DRL_SPS_CCRA_SFL_UAV}, while the CCRA is fixed at $\alpha _k = \beta_k=1/K$, $\forall k$.
\item $u={\widetilde{u}}$: The CCRA is decided by Algorithm \ref{alg_DRL_SPS_CCRA_SFL_UAV}, while $u$ is fixed at ${\widetilde{u}}$.
\end{itemize}

\begin{figure}[t]
\centering
\subfloat[Latency for one round. \label{time_round_resource_allocation}]
{\includegraphics[width=0.8\linewidth]{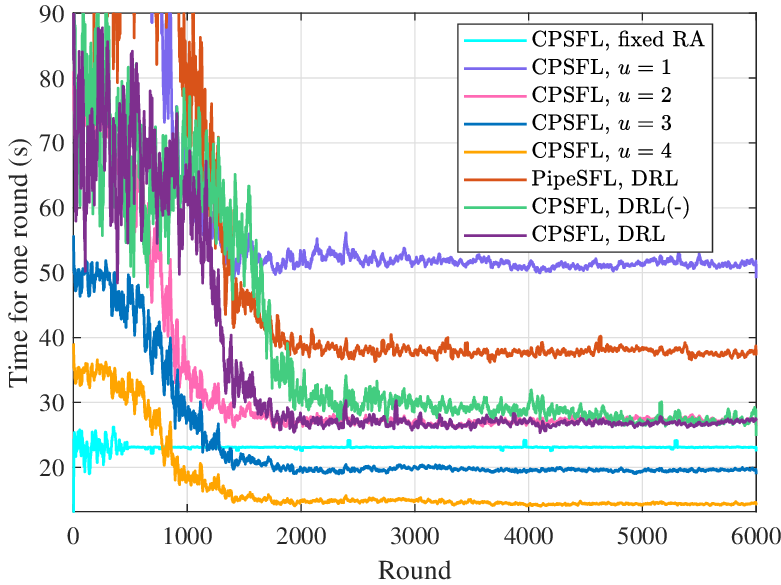}}
\vspace{-0.2cm}
\\
\subfloat[Maximum energy consumption of all UAVs. \label{energy_max_round_resource_allocation}]
{\includegraphics[width=0.8\linewidth]{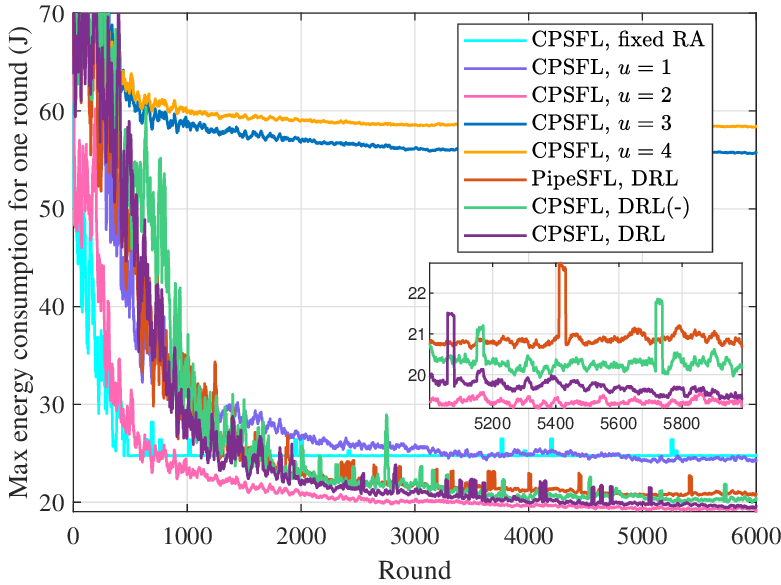}}
\caption{Latency and the maximum energy consumption of all UAVs for one round obtained by the DRL-based CPSFL scheme and the benchmarks: $I=5$.}
\label{latency_energy_resource_allocation} 
\vspace{-0.3cm}
\end{figure}
Fig. \ref{latency_energy_resource_allocation} shows the per-round latency $\tau\left( n \right)$ and the maximum energy consumption of all UAVs for one round ${\max}_k \!\left\{ {{e_k}\!\left( n \right)} \right\}$ obtained by the DRL-based CPSFL scheme and the benchmarks when $I=5$.
Each value is a moving average with a span of 21 rounds.
The proposed scheme converges faster than the DRL(-) scheme, attaining an average objective value (see \eqref{Problem_ori_obj}) of 105.8 over the last 1000 rounds, compared to 108.6 for the DRL(-) scheme.
This improvement stems from its ability to reduce environmental uncertainty by learning the relationship among UAV trajectories, actions, and rewards.
Moreover, the proposed DRL scheme can effectively handles the challenging joint optimization of SPS and CCRA within CPSFL.
For comparison, the fixed CCRA scheme yields an average objective value of 122.3, confirming that the proposed scheme achieves a superior latency-energy trade-off by adapting to heterogeneous UAV transmit powers and time-varying channel conditions induced by mobility.
Furthermore, the performance of the proposed scheme approaches that of the CPSFL scheme with $u=2$, which is the best fixed split point scheme evaluated by the objective values in \eqref{Problem_ori_obj}.
This shows that the BS agent in the proposed scheme can automatically select a high-performance split point.
Notably, when both CPSFL and PipeSFL employ the DRL for decision-making, CPSFL consistently achieves lower latency and energy consumption, further validating the benefits of sequential GT and CPSFL's two key enhancements.

\section{Conclusions}
\label{sec_Conclusions}
In this paper, we have proposed a sequential GT paradigm, where the server dedicates all downlink resources for the current GT.
We have further proposed CPSFL, characterized by downlink GT priority scheduling and intra-round asynchronous training.
We have investigated CPSFL-based LoRA fine-tuning of FMs in UAV networks and have formulated an optimization problem to minimize the per-round training latency and the worst-case client energy consumption by optimizing the SPS, the uplink bandwidth allocation, and the server computing frequency allocation.
To solve this problem, we have developed an attention-based DRL framework, where the BS agent decides the split point and the CCRA in each round by leveraging previous round information including UAV trajectories.
Simulation results have shown that the proposed DRL-based CPSFL scheme outperforms the benchmarks, the ablation variants, the fixed CCRA scheme, and approaches the best fixed-SPS scheme.
This study lays the foundation for finer-grained pipelining paradigms in SFL.

\appendices   
\section{Proof of theorem 1}
\label{sec_Proof_theorem1}
With Assumption \ref{assumption_channel_constant_round}, the lags satisfy $l_1\le l_2\le\cdots\le l_K$.
The proposed priority scheduling strictly follows this lag ordering. 
We suppose that there exist indices $m,k\in {\cal K}$ with $m>k$ (so $l_m\geq l_k$) such that swapping their priorities, i.e., assigning client $m$ a lower priority than client $k$, reduces the per-iteration training latency in \eqref{time_per_iteration}.
This latency is rewritten as
\begin{align}
T_{\rm{iter}}=&\max \bigg\{\max_{j\in\left[1,k-1\right]\cup\left[m+1,K\right]}{\left\{t_j+l_j\right\}},t_m+l_m,\nonumber\\
&\max_{q\in\left[k+1,m-1\right]}{\left\{t_q+l_q\right\}},t_k+l_k\bigg\}.
\end{align}
For any $q\in[k+1,m-1]$, we have $l_m\geq l_q\geq l_k$.
We denote $t_k=\sum_{j=k}^{K}\tau_{SG,j}$, thus we have $t_k>t_q>t_m$.
We use the superscript * to denote all intervals and timestamps after the swap.
After swapping, we have $t_m^\ast>t_q^\ast>t_k^\ast$, $t_m^\ast=t_k$, and $t_j^\ast=t_j$, $\forall j\in\left[1,k-1\right]\cup\left[m+1,K\right]$.
Besides, the per-iteration training latency becomes
\begin{align}
T_{\rm{iter}}^\ast=&\max \bigg\{\max_{j\in\left[1,k-1\right]\cup\left[m+1,K\right]}{\left\{t_j^\ast+l_j\right\}},t_m^\ast+l_m,\nonumber\\
&\max_{q\in\left[k+1,m-1\right]}{\left\{t_q^\ast+l_q\right\}},t_k^\ast+l_k\bigg\}.
\end{align}

Since $t_m^\ast=t_k$ and $l_m\geq l_k$, we have $t_m^\ast+l_m\geq t_k+l_k$.
Since $t_m^\ast>t_q$ and $l_m\geq l_q$, we have $t_m^\ast+l_m>t_q+l_q$. 
Since $t_m^\ast>t_m$, we have $t_m^\ast+l_m>t_m+l_m$.
Therefore, $T_{\rm{iter}}^\ast\geq T_{\rm{iter}}$, meaning the swap does not reduce latency, contradicting the initial assumption.
Moreover, any priority ordering can be obtained by starting from the descending-lag order and performing a sequence of such swaps. 
Since none of these swaps can reduce the latency, the ordering that prioritizes clients with larger lags is optimal.

\section{Proof of theorem 2}
\label{sec_Proof_theorem2}
With Assumption \ref{assumption_channel_constant_round}, the lags satisfy $l_1\le l_2\le\cdots\le l_K$.
The proposed priority scheduling strictly follows this lag ordering. 
We suppose that there exist indices $m,k\in {\cal K}$ with $m>k$ (so $l_m\geq l_k$) such that swapping their priorities, i.e., assigning client $m$ a lower priority than client $k$, reduces the per-round training latency.

We use the superscript * to denote all intervals and timestamps after the swap.
Let $T_{CPSFL,k}$ be the latency for client $k$ to complete all iterations in a round.
We will show $T_{CPSFL}^\ast\geq T_{CPSFL}$, thereby proving Theorem \ref{optimal_schedule_asynchronous}.

Firstly, lowering client $m$'s priority may cause its GT to wait for those of clients $k$ through $m-1$, so $T_{CPSFL,m}^\ast \geq T_{CPSFL,m}$.
Secondly, raising client $k$'s priority eliminates its need to wait for transmissions of clients $k+1$ through $m$, so $T_{CPSFL,k} \geq T_{CPSFL,k}^\ast$.

Thirdly, client $m$'s priority after the swap equals client $k$'s priority before the swap.
Thus, the number of higher-priority clients remains unchanged.  
Let $\varrho_k$ and $\varrho_k^\ast$ denote the number of GTs client $k$ must wait for before and after the swap, respectively, with $\varrho_k,\varrho_k^\ast\in\left[I,\sum_{j=k}^{K}I\right]$.  
Let ${\overline{\tau}}_{SG,k}$ and ${\overline{\tau}}_{SG,k}^\ast$ be the corresponding average per-transmission latencies.  
Then, we approximate $T_{CPSFL,k}\approx\varrho_k{\overline{\tau}}_{SG,k}+Il_k$ and $T_{CPSFL,m}^\ast\approx\varrho_m^\ast{\overline{\tau}}_{SG,m}^\ast+Il_m$.
Under Assumption \ref{assumption_latency_SG_equal}, we have $\tau_{SG,k}=\tau_{SG}$, $\forall k$ and ${\bar{\tau}}_{SG,m}^\ast={\bar{\tau}}_{SG,k}$.
Moreover, under Assumption \ref{assumption_latency_SG_equal} and $l_m\geq l_k$, client $k$ is enqueued more frequently than client $m$, implying $\varrho_m^\ast\geq\varrho_k$. 
Hence, $T_{CPSFL,m}^\ast\geq T_{CPSFL,k}$.
Similarly, the priority of client $j$ remains unchanged by the swap, $\forall j\in\left[k+1,m-1\right]$.
We approximate $T_{CPSFL,j}\approx\varrho_j{\bar{\tau}}_{SG,j}+Il_j$ and $T_{CPSFL,j}^\ast\approx\varrho_j^\ast{\bar{\tau}}_{SG,j}^\ast+Il_j$.
Under Assumption \ref{assumption_latency_SG_equal}, we have ${\bar{\tau}}_{SG,j}={\bar{\tau}}_{SG,j}^\ast$.
Moreover, under Assumption \ref{assumption_latency_SG_equal} and  $l_m\geq l_k$, client $k$ is enqueued more frequently than client $m$, implying $\varrho_j^\ast\geq\varrho_j$. 
Hence, $T_{CPSFL,j}^\ast\geq T_{CPSFL,j}$.

Finally, since client $j$'s priority are unaffected by the swap, $\forall j\in\left[1,k-1\right]\cup\left[m+1,K\right]$, their per-round training latency is basically unchanged, i.e., $T_{CPSFL,j}^\ast=T_{CPSFL,j}$.

These results imply $T_{CPSFL}^\ast\geq T_{CPSFL}$, contradicting the assumption that the swap reduces latency.  
Since any priority ordering can be reached via successive swaps from the descending-lag order, without reducing latency, the proposed lag-based priority scheduling is optimal.

\bibliographystyle{IEEEtran}%

\bibliography{IEEEabrv,papers_zzz}
\end{document}